\DocumentMetadata{
	lang		= en-US,
	pdfversion  = 1.7,
	pdfstandard = a-2b,
}

\documentclass[11pt]{article}
\pdfoutput=1
\usepackage{fullpage}
\usepackage[shortlabels]{enumitem}

\usepackage{array}%
\usepackage{lmodern}%
\usepackage{dcolumn}%
  \newcolumntype{d}[1]{D{.}{.}{#1}}%

\usepackage{longtable}%

\usepackage{subfiles}%

\usepackage{multirow}              %

\usepackage{nicematrix} %

\usepackage{microtype}%

\usepackage{graphicx} %
\usepackage{hyperref, xcolor} %

\hypersetup{
    pdfpagemode=UseOutlines,
    linkbordercolor=[rgb]{1,1,1},
    filebordercolor=[rgb]{1,1,1},
    citebordercolor=[rgb]{1,1,1},
    menubordercolor=[rgb]{1,1,1},
    urlbordercolor=[rgb]{1,1,1},
    colorlinks=true,
    linkcolor=[rgb]{0,0.2,0.5},
    citecolor=[rgb]{0,0.35,0},
    urlcolor=[rgb]{0.6,0.2,0.4},
    breaklinks=true
}

\usepackage{import} %

\usepackage{mathtools, amssymb, amsthm} %
\usepackage{mathrsfs} %

\usepackage{cleveref} %

\numberwithin{equation}{section}
\theoremstyle{plain}  %
\newtheorem{thm}{Theorem}[section]
\newtheorem*{thm*}{Theorem}
\newtheorem{lem}[thm]{Lemma}
\newtheorem*{lem*}{Lemma}

\newtheorem*{prop*}{Proposition}
\newtheorem{cor}[thm]{Corollary}
\newtheorem*{cor*}{Corollary}
\newtheorem{defn}[thm]{Definition}
\newtheorem*{defn*}{Definition}

\newtheorem*{claim*}{Claim}

\newtheorem*{question*}{Question}
\newtheorem{remark}[thm]{Remark}
\newtheorem*{remark*}{Remark}

\usepackage{braket} %
\usepackage{bbm} %

\usepackage[backend=biber,style=alphabetic,giveninits=false,maxbibnames=10,sorting=none]{biblatex}

\newif\ifdebug
\debugfalse

\ifdebug
\usepackage[left]{showlabels}

\fi

\ifdebug
\newcommand{\TODO}[1]{\ifdebug \textcolor{red}{TODO: (#1) \marginpar{\small\textcolor{red}{to-do}}} \fi}
\newcommand{\anote}[1]{\ifdebug\textcolor{green}{\small (Anand: #1)}\fi}
\newcommand{\dnote}[1]{\ifdebug\textcolor{teal}{\small (David: #1)}\fi}
\else
\newcommand{\TODO}[1]{}
\newcommand{\anote}[1]{}
\newcommand{\dnote}[1]{}
\fi

\def\be#1\ee{\begin{equation}#1\end{equation}}
\def\bea#1\eea{\begin{eqnarray}#1\end{eqnarray}}
\def\bearr#1\eearr{\begin{equation}\begin{aligned}#1\end{aligned}\end{equation}}

\newcommand{\NN}{\mathbb{N}}

\newcommand{\CC}{\mathbb{C}}

\DeclareMathOperator*{\E}{\mathbb{E}}
\newcommand{\mcH}{\mathcal{H}}

\newcommand{\qq}[1]{``#1''}

\newcommand{\norm}[1]{\left\lVert#1\right\rVert}
\newcommand{\brac}[1]{\left(#1\right)}
\newcommand{\sbrac}[1]{\left[#1\right]}
\newcommand{\cbrac}[1]{\left\{#1\right\}}

\newcommand{\abrac}[1]{\left\langle#1\right\rangle}
\newcommand{\abs}[1]{\left| #1 \right|}

\newcommand{\inprod}[2]{\left\langle#1,#2\right\rangle}

\newcommand{\eps}{\varepsilon}

\newcommand{\II}{\mathbbm{1}}
\renewcommand{\vec}[1]{\mathbf{#1}}
\newcommand{\mat}[1]{\mathbf{#1}}

\newcommand{\innerprod}[2]{\abrac{#1, #2}}

\newcommand{\Enc}{\mathsf{Enc}}
\newcommand{\Dec}{\mathsf{Dec}}
\newcommand{\negl}{\mathrm{negl}}
\newcommand{\blockmat}[4]{
    \left(
    \begin{array}{c|c}
        #1 & #2 \\
        \hline
        #3 & #4
    \end{array}
    \right)
}

\def\nice{nice} %
\def\niceness{niceness}

\newcommand{\alpbt}[1]{\mathcal{#1}}

\newcommand{\strat}[1]{\mathcal{#1}}
\newcommand{\game}[1]{\mathcal{#1}}
\newcommand{\gp}[1]{P_\game{#1}}
\newcommand{\compgval}[2]{\omega_{\mathsf{comp}}(\strat{#1}, \game{#2})}

\newcommand{\psE}[1]{\Tilde{\E}_{\strat{S}}\sbrac{#1}}
\newcommand{\QHE}{{\mathsf{QHE}}}
\newcommand{\approxQHE}{\underset{\negl_\QHE}{\approx}}

\newcommand{\Apovms}{\cbrac{M_{ax}}_{a\in\alpbt{A}, x\in\alpbt{X}}}
\newcommand{\Bpovms}{\cbrac{N_{by}}_{b\in\alpbt{B}, y\in\alpbt{Y}}}

\newcommand{\polyB}{p}

\newcommand{\Bthree}{\mathcal{B}_3}

\begin{document}

\title{A convergent sum-of-squares hierarchy for compiled nonlocal games}

\author{David Cui\texorpdfstring{\footnote{MIT, \href{mailto:dzcui@mit.edu}{dzcui@mit.edu}}}{} 
\and
Chirag Falor\texorpdfstring{\footnote{MIT, \href{mailto:cfalor@mit.edu}{cfalor@mit.edu}}}{}
\and
Anand Natarajan\texorpdfstring{\footnote{MIT, \href{mailto:anandn@mit.edu}{anandn@mit.edu}}}{}
\and
Tina Zhang\texorpdfstring{\footnote{MIT, \href{mailto:tinaz@mit.edu}{tinaz@mit.edu}}}{}}

\date{July 23, 2025}

\maketitle

\begin{abstract}
We continue the line of work initiated by Kalai et al. (STOC' 23), studying ``compiled" nonlocal games played between a classical verifier and a \emph{single} quantum prover, with cryptography simulating the spatial separation between the players. The central open question in this area is to understand the soundness of this compiler against quantum strategies, and apart from results for specific games, all that is known is the recent ``qualitative'' result of Kulpe et al. (STOC '25) showing that the success probability of a quantum prover in the compiled game is bounded by the game's quantum commuting-operator value in the limit as the cryptographic security parameter goes to infinity.
In this work, we make progress towards a \emph{quantitative} understanding of quantum soundness for general games, by giving a concrete framework to bound the quantum value of compiled nonlocal games. Building on the result of Kulpe et al. together with the notion of ``nice'' sum-of-squares certificates, introduced by Natarajan and Zhang (FOCS' 23) to bound the value of the compiled CHSH game, we extend the niceness framework and construct a hierarchy of semidefinite programs that searches exclusively over nice certificates. We show that this hierarchy converges to the optimal quantum value of the game. Additionally, we present a transformation to make any degree-1 sum-of-squares certificate nice. This approach provides a systematic method to reproduce all known bounds for special classes of games together with Kulpe et al.'s bound for general games from the same framework.
\end{abstract}

\ifdebug
{\color{red}
NOTATION:
\begin{itemize}
    \item QUESTION SETS: $\mathcal{X}, \mathcal{Y}$. ANSWER SETS: $\mathcal{A}, \mathcal{B}$.
    \item Alice PVMs/POVMs: $A_{ax}$, Bob: $B_{by}$ (NOTE THE ORDERING: answer first/ question second)
    \item ABSTRACT Alice: $M_{ax}$, ABSTRACT Bob: $N_{by}$
    \item optimal values $\omega_{set}^{\ast}(\mathcal{G})$. for a particular strategy of a particular correlation class $set$: $\omega_{set}(\mathcal{S}, \mathcal{G})$
\end{itemize}
}
\fi

\clearpage

\tableofcontents

\clearpage
\section{Introduction}
In a nonlocal game, a single verifier interacts with two untrusted provers who are not allowed to communicate. These games were used by Bell in 1964 to demonstrate nonlocality by comparing the performance of quantum provers (provers allowed to share entanglement) with that of classical provers \cite{Bell}. Since then, nonlocal games have been extensively studied, revealing striking properties---such as the existence of nonlocal games where a high probability of success necessitates the use of specific quantum measurements and states, a phenomenon known as \emph{self-testing} \cite{10.5555/2011827.2011830,Supic2020selftestingof}. As a result, nonlocal games have been particularly effective as components in multi-prover interactive proofs, enabling unconditionally secure protocols for certified randomness generation \cite{VV12} and verifiable delegation of quantum computation \cite{RUV13, Gri19, Col+19}.

In the single-prover setting, such protocols also exist, but their security relies on certain cryptographic assumptions \cite{mahadev2018classicalverification, Bra+18, MV21}. One may then ask whether it is possible to generically transform information-theoretically secure multi-prover interactive protocols to cryptographically secure \emph{single}-prover interactive protocols. To this end, Kalai et al. proposed a procedure that \emph{compiles} any nonlocal game into a single-prover interactive game \cite{KLVY}. This approach has found success in protocols for randomness generation \cite{10.1007/978-3-031-38554-4_6}, delegation of quantum computation \cite{NZ23, MNZ24}, and self-testing \cite{Bar+24, 2406.04986}.

A central step in analyzing these protocols is to bound the quantum value of the compiled game.
To date, all known bounds have been obtained ad hoc. The most successful
approach to bounding the compiled quantum value thus far has been the
sum-of-squares (SoS) approach. In the study of (not necessarily
compiled) nonlocal games, this is a standard approach to
upper-bounding the quantum value \cite{0803.4290, 0803.4373}, which is
in general undecidable \cite{MIP*}. In \cite{NZ23, Cui+24, Bar+24,
  2406.04986}, the approach was to find SoS decompositions with
particular algebraic properties, called \emph{nice}, suitable for
showing that the compilation procedure preserves the quantum
value. However, a priori, one should not expect that nice SoS
decompositions exist for \emph{all} nonlocal games. 

Remarkably, \cite{2408.06711} recently showed an asymptotic result for the quantum value of \emph{all} compiled nonlocal games. Their approach deviates from the SoS approach as they directly show that from an infinite family of strategies for the compiled nonlocal game, one can construct a strategy of matching value for the original nonlocal game. Specifically, they showed the following informal theorem:

\begin{thm*}[Informal, \cite{2408.06711}]
    Let $\mathcal{G}$ be any nonlocal game and $\mathcal{G}_{comp}$ be the corresponding compiled game. Then any prover running in time polynomial in the security parameter $\lambda$ wins $\mathcal{G}_{comp}$ with probability at most $\omega_{qc}^{\ast}(\mathcal{G}) + f(\lambda)$, where $\omega_{qc}^{\ast}(\mathcal{G})$ is the quantum commuting value of the original nonlocal game $\mathcal{G}$ and $f$ is some function that tends to $0$ as $\lambda \to \infty$.
\end{thm*}

Although this is a general result, applying to all nonlocal games, it is not useful for cryptographic applications since it tells us nothing about the rate at which $f$ tends to zero as $\lambda \to \infty$. Ideally, in order to implement a cryptographic protocol in practice, we would like to be able to choose a concrete value for the security parameter $\lambda$ based on our tolerances for the security of the scheme. Moreover, at the level of asymptotics, in cryptographic applications one is usually interested in understanding the behavior of the compilation of a \emph{family} of games indexed by a size parameter $n$, and it is important to know how the security of the protocol scales with $\lambda$ and $n$ in order to know that the protocol truly runs in polynomial time, since the value of $\lambda$ controls the computational resources required to implement the cryptography.

In all previous results \cite{NZ23, 10.1007/978-3-031-38554-4_6, Cui+24, Bar+24, 2406.04986} analyzing specific families of nonlocal games, such a handle on the convergence rate had been established, i.e., one has the following ``gold-standard'' theorem:

\begin{thm*}[Informal]
    Let $\mathcal{G}$ be a nonlocal game and $\mathcal{G}_{comp}$ be the corresponding compiled game. Then any prover running in time polynomial in the security parameter $\lambda$ wins $\mathcal{G}_{comp}$ with probability at most $\omega_{qc}^{\ast}(\mathcal{G}) + \mathrm{negl}(\lambda)$, where $\omega_{qc}^{\ast}(\mathcal{G})$ is the quantum commuting value of the original nonlocal game $\mathcal{G}$.
\end{thm*}

Although this theorem is stated for individual games rather than families of games, results of this form have been proven useful in handling families as well. For instance, in the argument systems for $\mathsf{BQP}$ and $\mathsf{QMA}$ constructed in \cite{NZ23}, results of the type given above for certain basic nonlocal games like the CHSH game were used to show that, to achieve a constant soundness gap overall in the argument system for instances of size $n$, it suffices to take $\lambda = n$, ensuring that the protocol overall runs in polynomial time.

Thus, in order to hope to use the KLVY compiler for cryptographic applications, we would like to solve the following question.

\begin{question*}[\cite{2408.06711}]
    Can we establish a quantitative bound for $f$ for general nonlocal games?
\end{question*}

In this work, we make progress towards a computational solution to this question in the SoS framework, recovering along the way the specific bounds of~\cite{NZ23, Cui+24, Bar+24, 2406.04986}.

\subsection{Main results}

\paragraph{Nice SoS decompositions for all games.} Our main result is establishing a convergent SDP hierarchy for the nonlocal game value that searches exclusively over nice SoS decompositions, which we call the ``one-sided NPA hierarchy.'' This, together with a generalized definition of niceness, establishes,

\begin{thm}[Informal]
Let $\game{G}$ be a nonlocal game with quantum commuting value $\omega_{qc}(\mathcal{G})$. Then there is a hierarchy of semidefinite programs (the ``one-sided NPA hierarchy'') that converges to $\omega_{qc}^{\ast}(\mathcal{G})$ from above, such that every sum-of-squares certificate from this hierarchy certifying an upper bound of $\omega'$ on the quantum value implies an upper bound of $\omega' + \negl(\lambda)$ on the compiled value, where the negligible function $\negl$ depends arbitrarily on the certificate.
\end{thm}

The principal advantage of our result is that it gives a systematic way to extract bounds for specific nonlocal games by using the SDP hierarchy to search over nice SoS certificates. In the case the SDP hierarchy converges at a finite level, the corresponding compiled game value would be bounded by $\omega_{qc}^{\ast}(\mathcal{G}) + \mathrm{negl}(\lambda)$. Thus, our result in a sense subsumes all previous bounds on specific games, since they all involved finding a nice SoS certificate in an ad-hoc manner.

We also note that the existence of the one-sided NPA hierarchy is interesting in the study of nonlocal games as well. Computationally, this is an implementable hierarchy which will always generate nice SoS decompositions which may simply self-testing arguments such as \cite{1911.01593}.

\paragraph{All degree-1 SoS certificates can be made nice without increasing the degree.}
Our second major result is about controlling the degree of nice SoS certificates. 
We show the equivalence of the original NPA hierarchy and the one-sided NPA hierarchy at the first level. For any certificate that lies in the first level of the NPA hierarchy, we show that there exists a nice certificate for the same bound in the first level. 

\begin{thm}[Informal]
    Let $\gp{G}$ be the game polynomial of a nonlocal game $\game{G}$. If we have a degree-$1$ NPA sum-of-squares certificate for $\omega - \gp{G}$, then we can construct a degree-$1$ nice sum-of-squares certificate for $\omega - \gp{G}$.
\end{thm}
This means we can bound the compiled value of all games which have an optimal quantum value in the first level of the NPA hierarchy. Binary XOR games are one such class of games. In \cite{Cui+24}, the authors constructed a nice certificate by exploiting some specific properties of binary XOR games \cite{Slofstra_2011}. This work reproduces this result using more generic techniques, and extends it to all games with degree-1 SoS certificates, regardless of the size of the answer alphabet. It is well known that XOR games are a very special class of nonlocal games (they are computationally tractable and have an exact SDP characterization), so we view our results as indicating the power of low-degree nice SoS certificates for general (non-XOR) games.

Finally, in \Cref{appsec:niceB3}, we give a nice SoS certificate for a $3$-answer generalization of the CHSH game, the $\Bthree$ game \cite{1911.01593}. This is an NPA level 2 game which gives us hope that the {\niceness} framework is more general than the results in this paper.

\subsection{Future work}

In this paper, we have laid down a general framework to bound the quantum value of compiled games. We conclude by outlining some directions for future research.

\paragraph{\texorpdfstring{NPA level-$k$ value to nice NPA level-$k$
value}{NPA level-k value to nice NPA level-k value}}
One promising direction for future work is to generalize the two proofs given in \Cref{chap:npa_level1}. As it stands, the argument in \Cref{sec:lvl1_npa_lvl1_nice} does not extend to NPA level‑2 because the level‑2 matrices $\mat{M}_A$ are not necessarily block‑diagonal. For instance, consider the $\Bthree$ game---a ternary generalization of CHSH with three possible answers for each prover \cite{1911.01593}. The $\Bthree$ game resides at the second level of the NPA hierarchy, and its existing SoS certificate is not ``nice''; it cannot be straightforwardly transformed into a nice certificate within the current framework. Nevertheless, as demonstrated in \Cref{appsec:niceB3}, we have constructed a nice level‑2 SoS certificate for the $\Bthree$ game. This suggests that a suitable generalization of our framework could systematically convert any non‑nice SoS certificate into a nice one at arbitrary levels of the NPA hierarchy. 

\paragraph{Interesting examples of NPA level-1 games}
We showed that every NPA level-1 game (i.e., a game whose commuting operator value is achieved at level-1 of the NPA hierarchy) admits a nice SoS certificate. However, aside from the binary XOR games, we currently know of no other nontrivial level-1 examples. Thus, it is natural to seek additional level-1 games. Identifying such games would expand the class of games whose compiled values we can bound via the SoS method and may shed light on which quantum properties are preserved under compilation.

\paragraph{Applications of the one-sided hierarchy to nonlocal games}
We are interested in whether the one‑sided NPA hierarchy introduced here can also benefit (not necessarily compiled) nonlocal games. A common approach to establishing self‑testing results is to use an SoS decomposition to derive algebraic relations that near‑optimal strategies must satisfy. Our one‑sided hierarchy provides a systematic way to search for highly structured SoS decompositions. Might this perspective streamline existing self‑testing proofs? For instance, in \cite{1911.01593}, a substantial amount of effort was devoted to deriving certain algebraic relations from the ones given in the SoS decomposition.

\paragraph{Broader vision}
In this work, we analyzed how the value of a nonlocal game behaves under compilation. Yet, in practical applications, the value alone is seldom the primary concern; instead, we are often interested in the game's self‑testing properties. These include questions such as whether the game admits a unique optimal strategy (in terms of the shared state and measurements) and whether strategies that achieve nearly optimal values are necessarily close to this optimal strategy. A broader goal of research on compiled nonlocal games is to understand how these properties transform under compilation. In particular, existing work \cite{NZ23} focuses on characterizing the ``Bob'' operators of a compiled strategy—can we likewise say something meaningful about the ``Alice'' operators (i.e., the prover's actions under homomorphic encryption) or about the prover's quantum state?

\subsection{Related work}

During the completion of this work, we became aware of a concurrent and related work by Klep et al. \cite{other_group} exploring similar ideas.

\subsection*{Acknowledgements}

AN acknowledges support from NSF CAREER Award No.~2339948. 

\section{Preliminaries} \label{chap:prelims}
In this section, we will introduce the basic concepts and notations used in this paper. 

\subsection{Cryptography}
We adopt several definitions from~\cite{NZ23}.

\begin{defn}
A \emph{QPT (quantum polynomial time) algorithm} is a logspace-uniform family of quantum circuits with size polynomial in the number of input qubits and in the security parameter.
If the circuits are unitary then we call this a \emph{(unitary) QPT circuit}.
A POVM~$\{M_\beta\}$ is called \emph{QPT-measurable} is there is a QPT circuit such that measuring some output qubits and post-processing gives rise to the same probabilities as the POVM.
A binary observable~$B$ is called \emph{QPT-measurable} if this is the case for the corresponding projective POVM.
This is equivalent (by uncomputation) to demanding that~$B$ interpreted as a unitary can be realized by a QPT circuit.
\end{defn}
Here we follow~\cite{NZ23} in considering security against uniform adversaries, and indeed all of our reductions are uniform.
We remark that we can also define security against non-uniform adversaries to obtain a stronger conclusion at the cost of relying on a stronger cryptographic assumption (specifically, QHE secure against non-uniform adversaries, which is quite standard in cryptography).
We now recall the notion of quantum homomorphic encryption (QHE). Our definition is modeled on that of~\cite{KLVY}, which includes the additional property of ``correctness with auxiliary input,'' which is necessary for the completeness of the KLVY compiler, and holds for known constructions of QHE.

\begin{defn}
    A \emph{quantum homomorphic encryption scheme} $\mathrm{QHE}=(\mathrm{Gen}, \mathrm{Enc}, \mathrm{Eval}, \mathrm{Dec})$ for a class of circuits $\mathcal{C}$ is defined as a tuple of algorithms with the following syntax.
    \begin{itemize}
        \item $\mathrm{Gen}$ is a PPT algorithm that takes as input the security parameter $1^\lambda$ and returns a secret key $\mathrm{sk}$.
        \item $\mathrm{Enc}$ is a PPT algorithm that takes as input a secret key $\mathrm{sk}$ and a classical input~$x$, and outputs a classical ciphertext~$c$.
        \item  $\mathrm{Eval}$ is a QPT algorithm that takes as input a tuple a classical description of a quantum circuit~$C: \mathcal{H} \times (\mathbb{C}^2)^{\otimes n} \to (\mathbb{C}^2)^{\otimes m}$, a quantum register with Hilbert space~$\mathcal{H}$, and a ciphertext~$c$, %
        and outputs a ciphertext~$\Tilde{c}$.
        If $C$ has classical output, we require that~$\mathrm{Eval}$ also has classical output.
        \item $\mathrm{Dec}$ is a QPT algorithm that takes as input a secret key $\mathrm{sk}$ and ciphertext $c$, and outputs a state $\ket{\psi}$.
        Additionally, if $c$ is a classical ciphertext, the decryption algorithm deterministically outputs a classical string $y$.
    \end{itemize}
    We require that the scheme satisfies the following properties.
    \begin{itemize}
        \item \emph{Correctness with Auxiliary Input:}
        For any $\lambda\in\mathbb{N}$, any quantum circuit $C: \mathcal{H}_A \times (\mathbb{C}^2)^{\otimes n} \to \{0,1\}^*$ with classical output, any state $\ket{\psi}_{AB} \in \mathcal{H}_A\otimes \mathcal{H}_B$, and any message $x\in\{0,1\}^n$, the following experiments output states with negligible trace distance:
        \begin{itemize}
            \item Game 1:
            Start with~$(x,\ket{\psi}_{AB})$.
            Apply the circuit~$C$, obtaining the the classical string~$y$.
            Return~$y$ and register~$B$.
            \item Game 2:
            Start with a key $\mathrm{sk}\leftarrow\mathrm{Gen}(1^\lambda)$, $c\in\mathrm{Enc}(\mathrm{sk}, x)$, and~$\ket{\psi}_{AB}$.
            Apply $\mathrm{Eval}$ with input~$C$, register~$A$, and ciphertext~$c$ to obtain~$\Tilde{c}$.
            Compute $\Tilde y \leftarrow \mathrm{Dec}(\mathrm{sk}, \Tilde{c})$.
            Return~$\Tilde y$ and register $B$.
        \end{itemize}
        \item \emph{CPA Security:} For all pairs of messages $(x_0, x_1)$ and any QPT adversary $\mathcal{A}$ it holds that
        \[
        \left| \Pr\left[1=\mathcal{A}(c_0)^{\mathrm{Enc}(\mathrm{sk}, \cdot)}\middle|
        \begin{array}{l}
             \mathrm{sk} \gets  \mathrm{Gen}(1^\lambda) \\
             c_0 \gets \mathrm{Enc}(\mathrm{sk}, x_0)
        \end{array}
        \right]-
        \Pr\left[1=\mathcal{A}(c_1)^{\mathrm{Enc}(\mathrm{sk}, \cdot)}\middle|
        \begin{array}{l}
             \mathrm{sk} \gets  \mathrm{Gen}(1^\lambda) \\
             c_1 \gets \mathrm{Enc}(\mathrm{sk}, x_1)
        \end{array}
        \right]
        \right| \leq \mathrm{negl}(\lambda),
        \]
    \end{itemize}
\end{defn}
\begin{remark}\label{remark:abuse of notation}
    In a slight abuse of notation, we often write expressions such as~$\mathop{\mathbb{E}}_{c\leftarrow\mathrm{Enc}(x)} f(\mathrm{Dec}(\alpha))$ as an abbreviation for an expectation value of the form~$\mathop{\mathbb{E}}_{\mathrm{sk} \gets \mathrm{Gen}(1^\lambda), c\leftarrow\mathrm{Enc}(\mathrm{sk}, x)} f(\mathrm{Dec}(\mathrm{sk}, \alpha))$.
\end{remark}

\subsection{Nonlocal games} \label{subsec:quantum_value_notation}

\begin{defn}
A \emph{nonlocal game} $\mathcal{G}$ is a tuple $(\mathcal{X}, \mathcal{Y}, \mathcal{A}, \mathcal{B}, \pi, V)$ consisting of finite sets~$\mathcal{X}$ and~$\mathcal{Y}$ of inputs for Alice and Bob, respectively, finite sets~$\mathcal{A}$ and~$\mathcal{B}$ of outputs for Alice and Bob, respectively, a probability distribution of the inputs~$\pi: \mathcal{X}\times \mathcal{Y}\to [0,1]$, and a verification function $V: \mathcal{A}\times \mathcal{B}\times\mathcal{X}\times \mathcal{Y} \to \{0,1\}$.
\end{defn}

A nonlocal game is played by a verifier and two provers, Alice and Bob.
In the game, the verifier samples a pair $(x,y) \leftarrow \pi$ and sends $x$ to Alice and $y$ to Bob.
Alice and Bob respond with $a\in \mathcal{A}$ and $b\in \mathcal{B}$, respectively.
They win if $V (x, y, a, b) = 1$.
The players are not allowed to communicate during the game, but they can agree on a strategy beforehand.
Their goal is to maximize their winning probability.
If we do not specify otherwise, the distribution $\pi$ will be the uniform distribution.

\subsubsection{Quantum tensor product and commuting operator strategies}

\begin{defn} \label{def:tensor_strat}
A \emph{quantum (tensor) strategy} $S$ for a nonlocal game $\mathcal{G}=(\mathcal{X}, \mathcal{Y}, \mathcal{A}, \mathcal{B}, \pi, V)$ is a tuple $S=(\mcH_A,\mcH_B,\ket{\psi}, \{A_{ax}\}, \{B_{by}\})$, consisting of finite-dimensional Hilbert spaces~$\mcH_A$ and~$\mcH_B$, a bipartite state $\ket{\psi}\in \mathcal{H}_A\otimes \mathcal{H}_B$, positive operator-valued measures (POVMs) $\{A_{ax}\}_{a\in \mathcal{A}}$ acting on $\mathcal{H}_A$ for each $x \in \mathcal{X}$ for Alice and POVMs $\{B_{by}\}_{b\in \mathcal{B}}$ acting on $\mathcal{H}_B$ for each $y \in \mathcal{Y}$ for Bob. Often we will drop the Hilbert spaces, and just write $S=(\ket{\psi}, \{A_{ax}\}, \{B_{by}\})$.
\end{defn}

For these families of strategies, we can without loss of generality, restrict to pure states and projective measurements (PVMs).
For a strategy $S=(\ket{\psi}, \{A_{ax}\}, \{B_{by}\})$, the probability of Alice and Bob answering $a, b$ when obtaining $x,y$ is given by $p(a,b|x,y)=\bra{\psi} A_{ax}\otimes B_{by}\ket{\psi}$. Therefore, the \emph{winning probability} of a quantum strategy~$S$ for the nonlocal game $\mathcal{G}$ is given by
\begin{align*}
    \omega_q(S,\mathcal{G})=\sum_{x,y} \pi(x,y) \sum_{a,b}V(a,b,x,y)p(a,b|x,y)=\sum_{x,y} \pi(x,y) \sum_{a,b}V(a,b,x,y)\bra{\psi} A_{ax}\otimes B_{by}\ket{\psi}.
\end{align*}
For a nonlocal game $\mathcal{G}$, we define the \emph{quantum value} $\omega_q^*(\mathcal{G})=\sup_{S}\omega_q(S,\mathcal{G})$ to be the supremum over all quantum tensor strategies for $\mathcal{G}$. A strategy $S$ is called \emph{optimal} for a game $\mathcal{G}$, if $\omega_q(S,\mathcal{G})=\omega_q^*(\mathcal{G})$.

The tensor-product structure is a way of mathematically representing the locality of the players employing a quantum strategy in a nonlocal game. However, there is a more general way to model this nonlocality mathematically.

\begin{defn}\label{commuting_strat}
    A \emph{commuting operator strategy} $\mathcal{S}$ for a nonlocal game $\mathcal{G}=(\mathcal{X}, \mathcal{Y}, \mathcal{A}, \mathcal{B}, \pi, V)$ is a tuple $\mathcal{S}=(\mcH,\ket{\psi}, \{A_{ax}\}, \{B_{by}\})$, consisting of a Hilbert space $\mcH$, a state $\ket{\psi}\in \mathcal{H}$, and two collections of mutually commuting POVMs $\{A_{ax}\}_{a\in \mathcal{A}}$ acting on $\mathcal{H}$ for each $x \in \mathcal{X}$ for Alice and POVMs $\{B_{by}\}_{b\in \mathcal{B}}$ acting on $\mathcal{H}$ for each $y \in \mathcal{Y}$ for Bob, i.e. $[A_{ax},B_{by}]=0$ for all $a,b,x,y\in \mathcal{A}\times \mathcal{B}\times \mathcal{X}\times \mathcal{Y}$. Like for quantum strategies, we will often omit the Hilbert space and write $\mathcal{S}=(\ket{\psi}, \{A_{ax}\}, \{B_{by}\})$ for a commuting operator strategy.
\end{defn}

Again, we can, without loss of generality, restrict to PVMs for this family of strategies. We can also write a similar expression for the winning probability of such a strategy:
\begin{equation}
    \omega_{qc}(\mathcal{S}, \mathcal{G}) = \sum_{x,y} \pi(x,y) \sum_{a,b}V(a,b,x,y)p(a,b|x,y)=\sum_{x,y} \pi(x,y) \sum_{a,b}V(a,b,x,y)\bra{\psi} A_{ax}  B_{by}\ket{\psi}, \label{eq:omegaqc-def}
\end{equation}
as well as define the \emph{commuting operator (also known as the quantum commuting) value} of a nonlocal game $\omega_{qc}^*(\mathcal{G})=\sup_{S}\omega_{qc}(\mathcal{S},\mathcal{G})$ to be the supremum over all commuting operator strategies~$\mathcal{S}$ for $\mathcal{G}$.

\subsubsection{Game algebras and representations}

For each nonlocal game $\mathcal{G}=(\mathcal{X}, \mathcal{Y}, \mathcal{A}, \mathcal{B}, \pi, V)$, we can associate a \emph{game algebra} which is a $C^{\ast}$-algebra that encodes all of the algebraic relations required to be satisfied by any commuting operator strategy. We begin with defining an algebra which encodes the relations forming a PVM:

\begin{defn}
    Given finite sets $\mathcal{X}$, $\mathcal{A}$, the \emph{PVM algebra} $\mathcal{A}_{PVM}^{\mathcal{A}, \mathcal{X}}$ is the universal $C^{\ast}$-algebra generated by orthogonal projectors $\cbrac{ M_{ax} }_{a \in \mathcal{A}, x \in \mathcal{X} }$ such that $\sum_{a \in \mathcal{A}} M_{ax} = I$ for all $x \in \mathcal{X}$. 
\end{defn}

It is not difficult to see that a collection of operators $\cbrac{A_{ax}}_{a \in \mathcal{A}, x \in \mathcal{X}}$ acting on $\mathcal{H}$ is a PVM if and only if these form a representation of $\mathcal{A}_{PVM}^{\mathcal{A}, \mathcal{X}}$. 

\begin{defn}
    For a nonlocal game $\mathcal{G}=(\mathcal{X}, \mathcal{Y}, \mathcal{A}, \mathcal{B}, \pi, V)$, we define the \emph{game algebra for $\mathcal{G}$} as $\mathcal{A}_{\mathcal{G}} := \mathcal{A}_{PVM}^{\mathcal{A}, \mathcal{X}} \otimes_{\mathrm{max}} \mathcal{A}_{PVM}^{\mathcal{B}, \mathcal{Y}}$. Moreover, this tensor product of PVM algebras is equal to the universal $C^{\ast}$-algebra generated by orthogonal projectors $\Apovms, \Bpovms$ such that $\sum_{a \in \mathcal{A}} M_{ax} = I$, $\sum_{b \in \mathcal{B}} M_{by} = I$ for all $x \in \mathcal{X}, y \in \mathcal{Y}$ and $[M_{ax}, N_{by}] = 0$ for all $a,b,x,y$.
\end{defn}

Note that we denote the abstract Alice and Bob operators as $M_{ax}, N_{by}$ and their representations as $A_{ax}, B_{by}$, respectively.

Representations of $\mathcal{A}_{\mathcal{G}}$ are exactly the operators in quantum commuting operator strategies. If we slightly rewrite our expression for the winning probability of a commuting operator strategy (Equation~\eqref{eq:omegaqc-def}), we obtain
\[ \omega_{qc}(\mathcal{S}, \mathcal{G}) = \bra{\psi} \left( \sum_{x,y} \pi(x,y) \sum_{a,b} V(a,b,x,y) A_{ax} B_{by} \right) \ket{\psi}.\]
In this expression, the data about the game $\mathcal{G}$ are encoded in the expression in parentheses, which we call the \emph{game polynomial $P_{\mathcal{G}}$ for the $\mathcal{G}$}:
\[ P_{\mathcal{G}} := \sum_{x, y} \pi(x,y) \sum_{a,b} V(a,b,x,y) M_{ax}N_{by} \in \mathcal{A}_{\mathcal{G}}. \]
In particular, the evaluation of $P_{\mathcal{G}}$ on a state and representation is exactly the winning probability for that particular strategy. We may sometimes abuse notation and call a polynomial which is proportional to $P_{\mathcal{G}}$ the game polynomial. 

\subsubsection{PVMs and observables} \label{sec:observables}

Given a (abstract) PVM $\cbrac{M_{ax}}_{a \in \mathcal{A}}$, with $d$ outcomes, i.e., $\mathcal{A} = [d]$, we can make a transformation to an \emph{observable}, which is a finite order unitary operator. This observable is described as
\[ M_{x} := \sum_{a} \omega_{d}^{a} M_{ax}, \]
where $\omega_{d}$ is the primitive $d$th root of unity. This has finite order $d$
\[ M_{x}^{d} = \brac{ \sum_{a} \omega^{a}_{d} M_{ax} }^{d} = \sum_{a} (\omega^{a}_{d})^{d} M_{ax} = \sum_{a} M_{ax} = I \]
and is unitary as
\[ M_{x}^{\dagger}M_{x} = \brac{ \sum_{a} \omega_{d}^{-a} M_{ax} }\brac{ \sum_{a^{\prime}} \omega_{d}^{a^{\prime}} M_{a^{\prime}x} } = \sum_{a} M_{ax} = I.  \]

Furthermore, given an order $d$ observable $M_{x}$, we can define
\[ M_{ax} := \sum_{j=0}^{d-1} \omega_{d}^{-a \cdot j} M_{x}^{j}. \]
This forms a PVM, and furthermore, these two transformations are inverses of each other, as can be seen by applying standard Fourier identities:
\begin{align*}
M_{ax}^\dagger &= \sum_{j} \omega_d^{+a\cdot j} M_x^{-j} = \sum_{j'} \omega_{d}^{-a \cdot j'} M_x^{j'} = M_{ax} \\
M_{ax}M_{a'x} &= \sum_{j,j'} \omega_d^{-aj - a'j'} M_{x}^{j + j'} \\
&= \sum_{k, j'} \omega_d^{-a(k -j') -a'j'} M_x^{k} \\
&= \sum_{k} \omega_d^{-ak} \left(\sum_{j'} \omega_d^{-(a' -a)j'}\right) M_x^k \\
&= \delta_{a,a'} \sum_k \omega_d^{-ak} M_x^k \\
&= \delta_{a,a'} M_{ax}.
\end{align*}

Because of this correspondence, we can transform between PVMs and observables. For example, given a game polynomial $P_{\mathcal{G}} = \sum_{a,b,x,y} c_{a,b,x,y} M_{ax}N_{by}$, we may write this as
\[ P_{\mathcal{G}} = \sum_{j,k, x,y} d_{j,k,x,y} M_{x}^{j} N_{y}^{k}, \]
for observables $\cbrac{M_{x}}$ and $\cbrac{N_{y}}$.

\subsubsection{Strongly non-signaling algebraic strategies}

In this section, we briefly review the correlation set defined in \cite{2408.06711} together with its equivalence to commuting operator correlations.

\begin{defn}[\cite{2408.06711}] \label{def:non-signaling-strategy}
    A strongly non-signaling algebraic strategy consists of a PVM algebra of Bob's operators $\mathcal{A}_{PVM}^{\mathcal{B}, \mathcal{Y}}$ and positive linear functionals $\phi_{ax} : \mathcal{A}_{PVM}^{\mathcal{B}, \mathcal{Y}} \to \CC$ for each $a\in\alpbt{A}, x\in\alpbt{X}$ such that there exists a ``global'' state $\phi: \mathcal{A}_{PVM}^{\mathcal{B}, \mathcal{Y}} \to \CC$ where
    \be \sum_{a\in \alpbt{A}} \phi_{ax} = \phi \ee
    for all $x \in \mathcal{X}$. Such a strategy gives rise to a correlation
    \be p(a,b|x,y) = \phi_{ax}(N_{by}), \ee
    where $N_{by}$ are the generators of $\mathcal{A}_{PVM}^{\mathcal{B}, \mathcal{Y}}$.
\end{defn}
In fact, \cite{2408.06711} use the POVM algebra to define these strategies, but we can use the PVM algebra without loss of generality.

As above, we can define the \emph{strongly non-signaling algebraic value} of a nonlocal game $\omega_{sns}^{\ast}(\mathcal{G}) = \sup_{S} \omega_{sns}(\mathcal{S}, \mathcal{G})$ to be the supremum over all strongly non-signaling algebraic strategies $\mathcal{S}$ for $\mathcal{G}$.

\cite{2408.06711} show that the quantum commuting value matches the strongly non-signaling algebraic value for any nonlocal game. In fact, they show that the correlation sets are equal:

\begin{thm}[\cite{2408.06711}] \label{lem:non-signaling-commuting}
    Let $\game{G}$ be a nonlocal game. A correlation can be induced by a non-signaling algebraic strategy if and only if it can be induced by a commuting operator strategy. So then,
\be \omega_{sns}^{\ast}(\game{G}) = \omega_{qc}^{\ast}(\game{G}). \ee
\end{thm}

\subsection{Compiled games} \label{sec:prelims-compiled-games}
\begin{defn}\label{def:compiledgame}
A \emph{compiled game} $\mathcal{G}_\mathrm{comp}$ consists of a nonlocal game $\mathcal{G}$ and a quantum homomorphic encryption scheme $\mathrm{QHE}=(\mathrm{Gen}, \mathrm{Enc}, \mathrm{Eval}, \mathrm{Dec})$. However, unlike a standard nonlocal game, it is played by a verifier and a single prover. The behaviour of the interaction is described as follows:
\begin{itemize}
    \item[1.] The verifier samples $(x,y)\leftarrow \pi$, $\mathrm{sk}\leftarrow \mathrm{Gen}(1^\lambda)$, and $c\leftarrow \mathrm{Enc}(\mathrm{sk},x)$. The verifier sends~$c$ to the prover.
    \item[2.] The prover replies with some classical ciphertext~$\alpha$.
    \item[3.] The verifier sends~$y$ (in the clear) to the prover.
    \item[4.] The prover replies with some classical message~$b$.
    \item[5.] The verifier computes $a:=\mathrm{Dec}(\mathrm{sk},\alpha)$ and accepts if and only if~$a\in \mathcal{A}$, $b\in \mathcal{B}$, and $V(a,b,x,y)=1$.
\end{itemize}
\end{defn}

\begin{defn}
A \emph{quantum strategy for a compiled game} is a tuple $(\mcH,\ket{\psi}, \{A_{\alpha c}\}, \{B_{by}\})$ consisting of a Hilbert space $\mcH$, an efficiently
(in QPT) preparable state $\ket{\psi}\in \mathcal{H}$, operators~$A_{\alpha c}=U_{\alpha c}P_{\alpha c}$ acting on~$\mathcal{H}$, where~$U_{\alpha c}$ is a QPT-measurable unitary and~$\{P_{\alpha c}\}_{\alpha\in \Lambda}$ (where $\Lambda$ is the set of outcomes), is a QPT-measurable PVM for all~$c\in \Enc(sk,x)$, $x\in \mathcal{X}$, as well as QPT-measurable PVMs~$\{B_{by}\}_{b\in \mathcal{B}}$ acting on~$\mathcal{H}$ for all~$y\in \mathcal{Y}$.
\end{defn}

For convenience, we let $\ket{\psi_{\alpha c}}:=A_{\alpha c}\ket{\psi}$ be the \emph{unnormalized post-measurement state} after Step~2 in \Cref{def:compiledgame}. For a strategy for the compiled game $S=(\ket{\psi}, \{A_{\alpha c}\}, \{B_{by}\})$, the probability of Alice and Bob answering $a, b$ after being given $x,y$ is denoted
\begin{align*}
    p(a,b|x,y)=\mathop{\mathbb{E}}_{c\leftarrow\mathrm{Enc}(x)}\sum_{\alpha; \mathrm{Dec}(\alpha)=a}\bra{\psi} (A_{\alpha c})^*B_{by}A_{\alpha c}\ket{\psi}
    =\mathop{\mathbb{E}}_{c\leftarrow\mathrm{Enc}(x)}\sum_{\alpha; \mathrm{Dec}(\alpha)=a}\bra{\psi_{\alpha c}}B_{by}\ket{\psi_{\alpha c}}
\end{align*}
It follows that the \emph{winning probability} of the quantum strategy~$S$ for the compiled game~$\mathcal{G}_\mathrm{comp}$ is given by
\begin{align}\label{eq:compiled winning prob}
    \omega_q(S,\mathcal{G}_\mathrm{comp})=\sum_{x,y} \pi(x,y) \sum_{a,b}V(a,b,x,y)\mathop{\mathbb{E}}_{c\leftarrow\mathrm{Enc}(x)}\sum_{\alpha; \mathrm{Dec}(\alpha)=a}\bra{\psi_{\alpha c}}B_{by}\ket{\psi_{\alpha c}}.
\end{align}

\begin{thm}\label{thm:KLVY}(\cite[Theorem 3.2]{KLVY})
If $\mathcal{G}_\mathrm{comp}$ is a compiled nonlocal game with underlying nonlocal game~$\mathcal{G}$. Then, there exists a compiled quantum strategy $\mathcal{S}$ for $\mathcal{G}_\mathrm{comp}$ and a negligible function $\eta(\lambda)$ such that
\begin{align*}
    \omega_q(\mathcal{S}, \mathcal{G}_\mathrm{comp})\geq \omega_{q}^{\ast}(\mathcal{G})-\eta(\lambda),
\end{align*}
where $\eta(\lambda)$ depends on $S$.
\end{thm}

\subsection{Block encodings and efficient measurement}

In analyzing strategies for compiled nonlocal games, a common move is to use the security property of the QHE scheme to argue that two different polynomials in the strategy operators have close expectation values on the state. In order to carry this out, \cite{NZ23} and subsequent works have used the formalism of \emph{block encodings} to show that polynomials in QPT-implementable measurement operators are themselves QPT-implementable. Below, we state and sketch the proof of the main theorem we will need for our analyses, without going into any detail on the block encoding formalism. For a more systematic treatment we refer the reader to~\cite{NZ23} and~\cite{Cui+24}.

\begin{thm}\label{thm:block-encoding}
    Let $\mathcal{G}_{\mathrm{comp}}$ be a compiled nonlocal game and let $S = (\mathcal{H}, \ket{\psi}, \{A_{\alpha c}\}, \{B_{by}\})$ be a quantum strategy for it. Let $w$ be any Hermitian polynomial in the operators $\{B_{by}\}$. Moreover, let $D_1, D_2$ be any two distributions over plaintext Alice questions that are sampleable in QPT. Then there exists a negligible function $\eta(\lambda)$ such that 
    \[ \left| \E_{x \gets D_1} \E_{c \gets \Enc(x)} \sum_{\alpha} \bra{\psi} A_{\alpha c}^* w A_{\alpha c} \ket{\psi} - \E_{x \gets D_2} \E_{c \gets \Enc(x)} \sum_{\alpha} \bra{\psi} A_{\alpha c}^* w A_{\alpha c} \ket{\psi} \right| \leq \eta(\lambda). \]
\end{thm}
\begin{proof}[Proof sketch]
    By Lemmas~2.17 and~2.18 of~\cite{Cui+24}, it follows that $w$ has a block encoding with scale factor $\Theta(1)$. Thus, there exists a QPT-implementable POVM $\{M_\beta\}_\beta$  that approximately measures $w$ by Lemma~2.20 of~\cite{Cui+24}. This implies the conclusion by the same argument as the proof of Lemma~2.21 of~\cite{Cui+24}.
\end{proof}

\subsection{The NPA and SoS hierarchies} \label{subsec:npa}
In this section, we will give a brief overview of the NPA hierarchy and its dual, the sum-of-squares hierarchy. 

\subsubsection{Sum-of-squares method} \label{subsec:sos}

A common method to upper-bound the quantum commuting value of nonlocal games is the \emph{sum-of-squares (SoS) method}. 

Let $\mathcal{G}$ be a nonlocal game with game polynomial $P_{\mathcal{G}}$. If we can write a SoS decomposition of the form
\be 
    \omega^{\prime} I - P_{\mathcal{G}} = \sum_i \lambda_i r_i^\dagger r_i \in \mathcal{A}_{\mathcal{G}}
\ee
where $\lambda_i$ are positive coefficients and $r_i, s_j$ are polynomials of PVM elements $\{M_{ax}\}_{a \in\alpbt{A}, x\in \alpbt{X}}$ and $\{N_{by}\}_{b \in \alpbt{B}, y\in \alpbt{Y}}$, then taking the expectation of the above equation with respect to the quantum state $\ket{\psi}$, we get
\[
    \omega^{\prime} - \bra{\psi} P_{\mathcal{G}} \ket{\psi} = \sum_i \underbrace{\lambda_i \bra{\psi} r_i^\dagger r_i \ket{\psi}}_\text{positive} \geq 0 \implies 
    \bra{\psi} P_{\mathcal{G}} \ket{\psi} \leq \omega^{\prime}.
\]
Hence, this method gives us a certificate that the value of the game is at most $\omega^{\prime}$.

\subsubsection{Moment matrix hierarchy (primal NPA hierarchy)}
The NPA hierarchy is a hierarchy of semidefinite programs that can be used to bound the quantum value of nonlocal games. The hierarchy was introduced by Navascues, Pironio, and Acin in \cite{0803.4290}. Each level $d$ of the hierarchy bounds on the quantum value by optimizing over possible states (positive-definite moment matrix) which reproduce the observed correlations between degree $d$ operators. 

Let the game polynomial be $\gp{G} = \sum_{a,b,x,y} c_{abxy} M_{ax} N_{by}$. For the $d$-th level of the NPA hierarchy, let $\vec{b}^d := \cbrac{I, M_{ax}, N_{by} \,:\, a,b,x,y}^{d} \subset \mathcal{A}_{\mathcal{G}}$ be a basis of monomials of degree at most $d$ in the PVM elements of Alice and Bob. Let $\Gamma^{d}$ be a \emph{moment matrix} with indices in $\mathbf{b}^{d}$. In particular, $\Gamma^{d}$ has dimension $\abs{\mathbf{b}^{d}} \times \abs{\mathbf{b}^{d}}$. Let $\mathbf{P}_{\mathcal{G}}$ be a $\abs{\mathbf{b}^{d}} \times \abs{\mathbf{b}^{d}}$ matrix such that $\innerprod{\mathbf{P}_{\mathcal{G}} }{\Gamma^{d}} = \sum_{a,b,x,y} c_{abxy} \Gamma^{d}_{M_{ax}, N_{by}}$.
The $d$-th level of the NPA hierarchy is defined as
\begin{equation} \label{eq:npa_formulation}
\begin{alignedat}{2}
\mathfrak{p}^{d}(\mathcal{G}) \quad = \quad & \max_{\mathclap{\Gamma^{d} \in \CC^{ \abs{\mathbf{b}^{d}} \times \abs{\mathbf{b}^{d}} } }} \qquad && \innerprod{ \mathbf{P}_{\mathcal{G}} }{\Gamma^d} \\
& \,\,\text{s.t.} \quad && \innerprod{\II_{I,I}}{\Gamma^d} = 1, \\
& && \innerprod{\mat{C}_k}{\Gamma^d} = 0, \quad \forall k \in \cbrac{1, \cdots, K} \\
& && \,\, \Gamma^d \succeq 0.
\end{alignedat}
\end{equation}

Here, $\mat{C}_k$ are constraint matrices which make sure the moments $\Gamma^d_{s_1, t_1} = \Gamma^d_{s_2, t_2}$, whenever $s_{1}^\dagger t_{1} = s_{2}^\dagger t_{2}$ for $s_{1}, t_{1}, s_{2}, t_{2} \in \mathbf{b}^{d}$. $\II_{I,I}$ is a matrix which is $1$ at the top-left corresponding to the elements $s=I, t=I$ and zero otherwise. The positivity constraint $\Gamma^d \succeq 0$ ensures that the moment matrix is positive semidefinite. The optimal value of the $d$-th level of the NPA hierarchy is denoted as $\mathfrak{p}^{d}(\mathcal{G})$. The NPA hierarchy converges to the quantum commuting value of the game, i.e., $\lim_{d \to \infty} \mathfrak{p}^{d}(\mathcal{G}) = \omega_{qc}(\mathcal{G})$ \cite{0803.4290}.\\

\subsubsection{Sum-of-squares hierarchy (dual NPA hierarchy)}
Instead of finding a moment matrix to maximize the expected value of the game polynomial, we can minimize $\nu$ such that $\nu I - \gp{G}$ has a SoS decomposition. Then, $\nu^*$ is an upper bound to the quantum value of the game as $\bra{\Psi} (\nu^* I - \gp{G} )\ket{\Psi} \geq 0 \implies \bra{\Psi} \gp{G} \ket{\Psi} \leq \nu^*$.\\

Formulating this procedure as an SDP problem, we get a hierarchy of semidefinite programs which is dual to the NPA hierarchy. This hierarchy of SoS certificates, parametrized by the degree-$d$ of the certificate, was formulated by \cite{0803.4373}. Mathematically, we can write the $d$-th level of the SoS hierarchy as

\begin{equation} \label{eq:dual_formulation}
    \begin{alignedat}{2}
        \mathfrak{d}^{d}(\mathcal{G}) \quad = \quad & \min_{ \mathclap{ \nu, \cbrac{y_k}_{k=1}^{K} } } \qquad && \nu \\
        & \,\,\text{s.t.} \quad && \nu \II_{1, 1} + \sum_{k=1}^{K} y_k \mat{C}_k - \mat{G}_p \succeq 0.
    \end{alignedat}
\end{equation}

Here, $\II_{1,1}, \mat{G}_p, \mat{C}_k$ have the same definition as in the NPA formulation \eqref{eq:npa_formulation} which are all defined for basis $\vec{b}^d$. Note that 
\bearr 
(\vec{b}^{d})^\dagger \II_{1,1} \vec{b}^{d} &= I,\\
(\vec{b}^{d})^\dagger \mat{G}_p \vec{b}^{d} &= \gp{G},\\
(\vec{b}^{d})^\dagger \mat{C}_k \vec{b}^{d} &= 0,
\eearr
in the algebra $\mathcal{A}_{\mathcal{G}}$, where here we are viewing $\mathbf{b}^{d}$ as a vector of monomials. Hence, solving this SDP and getting values $\mathfrak{d}^{d}(\mathcal{G})$ and $\cbrac{y_{k}}_{k =1}^{K}$ gives us
\[ \mathfrak{d}^{d}(\mathcal{G}) \II_{1, 1} + \sum_{k=1}^{K} y_k \mat{C}_k - \mat{G}_p = \sum_{i} \lambda_{i} \Pi_{i}^{\dagger}\Pi_{i} \succeq 0, \]
where the equality follows from the spectral decomposition. Then, evaluating both sides on $\mathbf{b}^{d}$ gives us the SoS decomposition as outlined in \Cref{subsec:sos}.

By duality, we get that $\mathfrak{p}^{d}(\mathcal{G}) \leq \mathfrak{d}^{d}(\mathcal{G})$ for all $d$. The SoS hierarchy converges to the quantum value of the game, i.e., $\lim_{d \to \infty} \mathfrak{d}^{d}(\mathcal{G}) = \omega_q(\mathcal{G})$ \cite{0803.4373}.

\section{Bounds for all compiled games through nice sum-of-squares decomposition} \label{chap:niceness}

In this section, building upon \cite{NZ23,Cui+24}, we will generalize the notion of {\niceness} of SoS decomposition. The SoS method gives us a framework to show upper bounds on the quantum value of nonlocal games, whilst nice SoS lets us show bounds on the value of \emph{compiled} nonlocal games. Later, we will define a convergent SDP hierarchy which will search exclusively over these nice SoS certificates, and we will show that this hierarchy also converges to the quantum value of the game. This gives us a hierarchy of upper bounds on the value of compiled nonlocal games.

\subsection{Generalizing the concept of nice sum-of-squares decomposition}

The quantum soundness of the KLVY compiler was proved for the CHSH game by \cite{NZ23}. Their proof relies on the \emph{niceness} of the SoS decomposition of the games. As done in \cite{NZ23}, one can define a \qq{pseudo-expectation} such that the pseudo-expectation of the game polynomial is negligibly far from the expectation of the compiled game.

Their paper defined the pseudo-expectation for Alice and Bob operators when both the questions and answers are binary. This definition was restricted to terms of at most degree 2, consisting of at most one Alice and one Bob observable. This pseudo-expectation was generalized for arbitrary monomials in $A_x, B_0, B_1,$ for a fixed $x \in \alpbt{X}$ by  \cite{2406.04986}. Their proof can be adapted to generalize this further for arbitrary monomials of the POVM elements $\{M_{ax}\}_{a \in\alpbt{A}}, \{N_{by}\}_{b \in \alpbt{B}, y \in \alpbt{Y}}$ still restricted to a fixed $x$. 

Firstly, let us define what a nice SoS decomposition is:
\begin{defn} \label{def:niceness}
    Let $G$ be a game polynomial. Assume that $G$ has the following sum-of-squares decomposition
    \begin{equation}
        G = \sum_{i=1}^{n} \lambda_i r_i^\dagger r_i + \sum_{j=1}^{m} \mu_j s_j,
    \end{equation}
    where $\lambda_i \geq 0$, $r_i$ are polynomials in the variables $M_{ax}, N_{by}$, and $s_j$ are constraint polynomials which should evaluate to $0$ for the game. We say that the SoS decomposition is \emph{nice} if each $r_i$ is a polynomial in the POVM elements $\cbrac{M_{a_i  x}}_{a_i \in \alpbt{A}}, \cbrac{N_{b_jy_k}}_{b_j \in \alpbt{B}, y_k \in \alpbt{Y}}$ for a fixed $x$. Specifically, each $r_i$ contains projectors corresponding to only one question $x$ of Alice.
\end{defn}

\subsubsection{Defining the pseudo-expectation for nice polynomials}

Let $\mathcal{G}$ be a nonlocal game and let $\strat{S} = (\ket{\psi}, \cbrac{M_{\alpha \chi}},\cbrac{N_{by}})$ be a strategy for its compilation. As in \Cref{sec:prelims-compiled-games} above, let 
\[ \ket{\Psi_{\alpha  \chi}} = M_{\alpha\chi} \ket{\Psi} \] 
be the state of the prover after the first round of the game. We will use this strategy to define a \emph{pseudo-expectation} operator mapping formal polynomials in the variables $M_{a\chi}$, $N_{by}$ to complex numbers. We denote pseudo-expectation by $\psE{\cdot}$ and define it as follows:
\begin{defn}
Treat the POVM elements of Alice and Bob $\cbrac{M_{a x}}_{a_i \in \alpbt{A}}, \cbrac{N_{by}}_{b \in \alpbt{B}, y \in \alpbt{Y}}$ as formal variables that follow the commutation relations and orthonormality relations: 
    \bearr \sbrac{M_{a x}, N_{by}} = 0 \\
    M_{a_i  x} M_{a_jx} = \begin{cases}
        M_{a_ix} \quad \text{if $a_i = a_j$}\\
        0 \quad \text{otherwise}
    \end{cases}\\
    N_{b_i  y} N_{b_j y} = \begin{cases}
        N_{b_i y} \quad \text{if $b_i = b_j$}\\
        0 \quad \text{otherwise}
    \end{cases}
    \eearr
    For a strategy $\strat{S} =  (\ket{\Psi}, \cbrac{M_{\alpha  \chi}},\cbrac{N_{b y}})$, we define the pseudo-expectation $\psE{\cdot}$ as an operator over formal polynomials of these variables, with the following properties:
    \begin{enumerate}
        \item $\psE{\cdot}$ is linear.
        \item $\psE{\II} = 1$.
    \end{enumerate}
    On monomials, we define the value of this pseudo-expectation as follows:
    \bearr &\psE{w_B\brac{\cbrac{N_{b_j y_k}}_{b_j \in \alpbt{B}, y_k \in \alpbt{Y}}}} \coloneq \E_{x \in \alpbt{X}}\E_{\chi: \mathsf{Enc}(x) = \chi} \sum_{a \in \alpbt{A}} \sum_{\alpha: \Dec\brac{\alpha} = a} \bra{\Psi_{\alpha  \chi}} w_B\brac{\cbrac{N_{b_j y_k}}_{b_j \in \alpbt{B}, y_k \in \alpbt{Y}}} \ket{\Psi_{\alpha  \chi}},\\
    &\psE{\brac{M_{a_i  x}} w_B\brac{\cbrac{N_{b_j y_k}}_{b_j \in \alpbt{B}, y_k \in \alpbt{Y}}}} \coloneq \E_{\chi: \mathsf{Enc}(x) = \chi} \sum_{\alpha: \Dec\brac{\alpha} = a_i} \bra{\Psi_{\alpha  \chi}}w_B\brac{\cbrac{N_{b_j y_k}}_{b_j \in \alpbt{B}, y_k \in \alpbt{Y}}} \ket{\Psi_{\alpha  \chi}}, \\
    &\psE{\brac{M_{a_i  x}M_{a_j  x}} w_B\brac{\cbrac{N_{b_j y_k}}_{b_j \in \alpbt{B}, y_k \in \alpbt{Y}}}} \coloneq 0 \quad\text{ if } a_i \neq a_j.
    \eearr
\end{defn}

This definition defines the pseudo-expectation over all nice SoS decompositions of the game polynomial. For any monomial in the decomposition, we can bring all the $M_{ax}$ terms to the left under the commutation relations $\sbrac{M_{a x}, N_{b y}} = 0$. Then, we can apply one of the above definitions to the monomial to calculate the pseudo-expectation. \\

One constraint of the POVM algebra, which is not ensured above, is the sum-to-one constraint $\sum_{a \in \alpbt{A}} M_{ax} = \II$. In Lemma \ref{lem:psE-sum-to-one-constraint}, we show that the pseudo-expectation nearly satisfies this constraint. 

\begin{lem} \label{lem:psE-sum-to-one-constraint}
    For any Hermitian polynomial $w_B\brac{\cbrac{N_{b_j y_k}}_{b_j \in \alpbt{B}, y_k \in \alpbt{Y}}}$, there exists a negligible function $\negl(\lambda)$ (possibly depending on $w_B$) such that
\be \abs{\psE{ (\II - \sum_{a \in \alpbt{A}} M_{a x} )w_B}} \leq \negl(\lambda).\ee
\end{lem}
\begin{proof}
   We can expand the pseudo-expectations as follows:
\begin{align}
    \abs{\psE{ (\II - \sum_{a \in \alpbt{A}} M_{a x} )w_B}} &= \abs{ \psE{w_B} - \psE{\sum_{a \in \alpbt{A}} M_{a x} w_B}}   \\
    &= \Bigg|
        \E_{x' \in \alpbt{X}}\E_{\chi \gets \Enc(x')} \sum_{a \in \alpbt{A}} \sum_{\alpha: \Dec\brac{\alpha} = a} \bra{\Psi_{\alpha  \chi}} w_B \ket{\Psi_{\alpha  \chi}}  \nonumber \\
          &\qquad  - \quad \E_{\chi \gets \Enc(x) } \sum_{a \in \alpbt{A}} \sum_{\alpha: \Dec\brac{\alpha} = a} \bra{\Psi_{\alpha  \chi}} w_B \ket{\Psi_{\alpha  \chi}} \Bigg| \\
          &=\Bigg|
        \E_{x' \in \alpbt{X}}\E_{\chi \gets \Enc(x')} \sum_{\alpha}\bra{\Psi_{\alpha  \chi}} w_B \ket{\Psi_{\alpha  \chi}}  \nonumber \\
          &\qquad  - \quad \E_{\chi \gets \Enc(x) } \sum_{\alpha} \bra{\Psi_{\alpha  \chi}} w_B \ket{\Psi_{\alpha  \chi}} \Bigg| \label{eq:xprime-x-close} 
\end{align}
Now, the two distributions over ciphertexts $\chi$ in the two terms in~\eqref{eq:xprime-x-close} are both efficiently sampleable. Hence, by applying Theorem~\ref{thm:block-encoding}, we conclude that
\be
\abs{\psE{ (\II - \sum_{a \in \alpbt{A}} M_{ax} )w_B}}  \leq \negl(\lambda).
\ee
\end{proof}

By linearity, we have defined the pseudo-expectation for all nice polynomial expressions in the basis. 

\subsubsection{Bounding the compiled value using the pseudo-expectation}

Now that we have defined the pseudo-expectation, we can now show the following lemma which will allow us to use the sum-of-squares method on the compiled version of games.

\begin{lem} \label{lem:psE-nice}
    Let $\cbrac{M_{a_i  x}}_{a_i \in \alpbt{A}}, \cbrac{N_{b_jy_k}}_{b_j \in \alpbt{B}, y_k \in \alpbt{Y}}$ be POVM projectors. Any polynomial in them can be written in the form
    \begin{equation}
        S = \polyB_\phi\brac{\cbrac{N_{b_j y_k}}_{b_j \in \alpbt{B}, y_k \in \alpbt{Y}}} + \sum_{a \in \alpbt{A}} \brac{M_{a  x}} \polyB_a\brac{\cbrac{N_{b_j y_k}}_{b_j \in \alpbt{B}, y_k \in \alpbt{Y}}}
    \end{equation}
    where $\polyB_a$ and $\polyB_\phi$ are complex polynomials in the Bob POVM elements. Then, the pseudo-expectation of $S^\dagger S$ is non-negative up to a negligible function. That is,
    \begin{equation}
        \psE{S^\dagger S} \geq -\negl\brac{\lambda}
    \end{equation}
    where $\negl$ is a negligible function of the security parameter $\lambda$ depending on the $\QHE$ scheme used in compilation, the strategy $\strat{S}$ and $S$.
\end{lem}

\begin{proof}
    The proof is essentially by direct calculation, and follows along
    the lines of \cite{2406.04986}, albeit adapted to the projector
    algebra and for multiple questions and multiple answers.
    
    Firstly, in preparation for using the relation $\sum_{a\in \alpbt{A}} M_{ax} = \II$, let us write
    \be 
        \polyB_\phi = \sum_{a\in \alpbt{A}} M_{ax} \polyB_\phi + \brac{\II - \sum_{a \in \alpbt{A}} M_{ax
        }}\polyB_\phi.
    \ee
     The pseudo-expectation operator $\psE{\cdot}$ does not exactly respect the relation $\sum_{a \in \alpbt{A}} M_{a x} = \II$, but rather only respects it up to a negligible error. Therefore, we will keep around the terms containing a factor of $(\II - \sum_{a \in \alpbt{A}} M_{ax})$ until we apply the pseudo-expectation operator.
    
    This means we can write $S$ as
    \be 
        S = \sum_{a\in \alpbt{A}} M_{a x} q_a + \brac{\II - \sum_{a \in \alpbt{A}} M_{a x}} \polyB_\phi,
    \ee
    where $q_a = \polyB_a + \polyB_\phi$.
    From orthogonality of $\cbrac{M_{a x}}_{a\in \alpbt{A}}$, we have:
    \begin{equation}
        M_{{a_1} x} M_{{a_2} x} = \begin{cases}
            M_{a x} & \text{if } a_1 = a_2 = a \\
            0 & \text{if } a_1 \neq a_2
        \end{cases}
    \end{equation}
    This implies that
    \be 
    \brac{\II - \sum_{a \in \alpbt{A}} M_{a x}}^2 = \II - \sum_{a \in \alpbt{A}} M_{a x}.
    \ee
    Applying the orthogonality to $S^\dagger S$, we get
    \bearr
        S^\dagger S &= \sum_{a_1} \sum_{a_2} M_{{a_1}  x} M_{{a_2}  x} q_{a_1}^\dagger q_{a_2} + \sum_{a_1} M_{{a_1}  x} \brac{1 - \sum_{a_2} M_{{a_2}  x}} (q_{a_1}^\dagger \polyB_\phi + \polyB_\phi^\dagger q_{a_1}) + \brac{\II - \sum_{a \in \alpbt{A}} M_{a x}}^2 \polyB_\phi^\dagger \polyB_\phi\\
        &= \sum_{a} M_{ax} q_a^\dagger q_a  + \sum_{a} \underbrace{M_{ax} \brac{1 - M_{ax}}}_{=0} (q_a^\dagger \polyB_\phi + \polyB_\phi^\dagger q_a) + \brac{\II - \sum_{a \in \alpbt{A}} M_{ax}} \polyB_\phi^\dagger \polyB_\phi\\
        &= \sum_{a} M_{ax} q_a^\dagger q_a + \brac{\II - \sum_{a \in \alpbt{A}} M_{ax}} \polyB_\phi^\dagger \polyB_\phi.
    \eearr
    where $M_{a x} \brac{1 - M_{ax}} = M_{a x} - M_{a x} = 0$, and $\polyB_\phi^\dagger \polyB_\phi$ is a polynomial in the Bob POVM elements.
    Now, we can apply the pseudo-expectation operator to $S^\dagger S$:
    \begin{align}
        \psE{S^\dagger S} &= \sum_{a} \psE{M_{a x} q_a^\dagger q_a} + \underbrace{\psE{\brac{\II - \sum_{a \in \alpbt{A}} M_{a x}} \polyB_\phi^\dagger \polyB_\phi}}_{\negl(\lambda)}\\
        &\approxQHE \sum_{a} \E_{\chi \gets \Enc(x) } \sum_{\alpha: \Dec\brac{\alpha} = a} \bra{\Psi_{\alpha  \chi}} q_a^\dagger q_a \ket{\Psi_{\alpha  \chi}},
    \end{align}
    where in passing to the second line we have applied Lemma~\ref{lem:psE-sum-to-one-constraint}, which we may do since $p_\phi^\dagger p_\phi$ is a Hermitian polynomial.

    Hence, we get
    \bearr
        \psE{S^\dagger S} &\approxQHE \E_{\chi \gets \Enc(x) } \sum_a  \sum_{\alpha: \Dec\brac{\alpha} = a}  \underbrace{\bra{\Psi_{\alpha  \chi}}  q_a^\dagger q_a \ket{\Psi_{\alpha  \chi}}}_{\geq 0} \\
        &\geq 0,
    \eearr
    where the inequality follows from the fact that the expectation of a square of a polynomial is always non-negative. 
   
    We have shown that $\psE{S^\dagger S} \approxQHE h$ for some $h \geq 0$. So, $\abs{\psE{S^\dagger S} - h} \leq \negl\brac{\lambda}$. So, we conclude that $\psE{S^\dagger S} \geq -\negl\brac{\lambda}$.
\end{proof}

The implication of Lemma \ref{lem:psE-nice} is the following: suppose we have a nice sum-of-squares certificate certifying an upper bound $\omega$ on the game polynomial, i.e. we have
\be \omega - \gp{G} = \sum_i \lambda_i r_i^\dagger r_i + \sum_j \mu_j s_j,\ee
with the polynomials $r_i$ satisfying the niceness condition as defined in definition \ref{def:niceness}. Then by applying the pseudo-expectation operator to both sides of this expression, and applying Lemma \ref{lem:psE-nice}, we can conclude that $\omega + \negl(\lambda)$ is an upper bound on the success probability of any \emph{compiled} strategy to the game. \\

This can be formalized as the following theorem. This corresponds to Corollary 4.5 in \cite{Cui+24}.
\begin{thm}\label{thm:nice-compiles}
    Let $\game{G}$ be a game with the game polynomial $\gp{G}$. If $\omega - \gp{G}$ has a nice SoS decomposition, then for any computationally bounded strategy $\strat{S}$, there exists a negligible function $\eta(\lambda)$ of the security parameter $\lambda$ such that
    \begin{equation}
        \compgval{S}{G} \leq \omega + \eta(\lambda),
    \end{equation}
    where $\compgval{S}{G}$ is the value achieved by strategy $\strat{S}$ in the compiled game.
\end{thm}
Note that the {\niceness} here is more general than the one defined in \cite{NZ23,Cui+24,2406.04986}. The definition here allows polynomials with terms consisting of different answers of Alice as long as they are for the same question.

\section{A hierarchy searching over nice SoS} \label{sec:one-sided-npa}

In this section, we will present another hierarchy of semidefinite programs, which we call the \emph{one-sided NPA hierarchy}. This hierarchy is similar to the NPA hierarchy in that we are searching over moment matrices indexed by restricted-degree operators of the game algebra. However, it differs in that Alice’s operators are always degree-1 while the degree of Bob’s operators increase with $d$. We will show that this restricted version of the NPA hierarchy also converges to the quantum value of the game. This version of the NPA hierarchy should be thought of as a convergent hierarchy characterizing strongly non-signaling algebraic correlations similar to how the original NPA hierarchy \cite{0803.4290} characterizes commuting operator correlations.

\subsection{Hierarchy over monomials of Bob's operators---the one-sided NPA hierarchy}

Inspired by \cite{0803.4290} and \cite{2408.06711}, we define a new hierarchy of semidefinite programs, which we call the \emph{one-sided NPA hierarchy}. To motivate the definition of the one-sided NPA hierarchy, let us take rewrite the NPA hierarchy defined in \Cref{eq:npa_formulation}.

First, notice that the moment matrix $\Gamma^{d}$ is really defining a linear functional on the subspace spanned by degree $2d$ monomials. In particular, if $\Gamma^{d}$ is some moment matrix then $\Gamma^{d}_{s, t} = \phi^{d}(s^{\dagger}t)$ for some linear functional~$\phi^{d} : (\mathcal{A}_{\mathcal{G}})_{\leq 2d} \to \CC$ due to the constraints imposed by the $\mathbf{C}_{k}$. Hence, we can rewrite the optimization problem as 

\begin{equation} \label{eq:npa_lin_formulation}
\begin{alignedat}{2}
\mathfrak{p}^{d}(\mathcal{G}) \quad = \qquad & \max_{\mathclap{ \phi^{d} : (\mathcal{A}_{\mathcal{G}})_{\leq 2d} \to \CC } } \qquad && \sum_{a,b,x,y} c_{abxy} \phi^{d}(M_{ax}N_{by}) \\
& \,\,\text{s.t.} \quad && \phi^{d}(I) = 1, \\
& && \phi^{d} \succeq 0,
\end{alignedat}
\end{equation}
where $\phi^{d} \succeq 0$ means that $\phi^{d}(s^{\dagger}s) \geq 0$ for all $s$ of degree less than or equal to $d$.

We now define:

\begin{defn}[One-sided NPA Hierarchy] \label{def:one-sided-npa-hierarchr}
Let the game polynomial be $P_{\mathcal{G}} = \sum_{a,b,x,y} c_{abxy} M_{ax} N_{by}$. The $d$-th level of the one-sided NPA hierarchy is defined as follows:

\begin{equation} \label{eq:new_npa_lin_formulation}
    \begin{alignedat}{3}
        \mathbbm{p}^{d}(\mathcal{G}) \quad = \qquad\qquad & \max_{\mathclap{ \cbrac{ \phi^{d}_{ax} : (\mathcal{A}_{PVM}^{\mathcal{B}, \mathcal{Y}})_{\leq 2d} \to \CC }_{a,x} } } \qquad\qquad && \sum_{a,b,x,y} c_{abxy} \phi^{d}_{ax}(N_{by}) && \\
        & \,\,\mathrm{s.t.} \quad && \sum_{a\in\alpbt{A}} \phi^d_{a0} = \sum_{a \in \mathcal{A}} \phi^{d}_{ax}, \quad \forall x \in\alpbt{X}, && \quad \text{(consistency)}\\
        & && \sum_{a \in \mathcal{A}} \phi^{d}_{a0}(I) = 1, &&\quad \text{(identity)} \\
        & && \phi^{d}_{ax} \succeq 0, \quad \forall a \in \mathcal{A}, x \in \mathcal{X}, &&
    \end{alignedat}
\end{equation}
where $\phi^d_{ax}$ is a linear functional defined on polynomials of Bob's operators of degree up to $2d$ and the identity constraint $\sum_{a \in \mathcal{A}} \phi^{d}_{a0}(I) = 1$ can be chosen to be any $x \in \mathcal{X}$ not just $0$ because of the consistency constraint.
\end{defn}
This hierarchy searches over strongly non-signaling algebraic strategies which are restricted to only degree $2d$ moments of Bob's operators. In section~\ref{sec:conv_oneNPA_hier}, we show that this hierarchy converges to the commuting value of the game. However, first, let's look at these semi-definite programs (SDP) in their standard form and show that the dual of this hierarchy is the nice SoS hierarchy which searches over degree-$d$ nice SoS certificates.

\paragraph{Standard form of the one-sided NPA hierarchy}

Let $\mathbf{s}^{d} = \cbrac{I, N_{by} \,:\, b \in \mathcal{B},y \in \mathcal{Y}}^{d}$ be the monomials of degree up to $d$ in $\mathcal{A}_{PVM}^{\mathcal{B}, \mathcal{Y}}$. Now, each $\phi^{d}_{ax} : (\mathcal{A}_{PVM}^{\mathcal{B}, \mathcal{Y}})_{\leq 2d} \to \CC$ can be viewed as a $\abs{\mathbf{s}^{d}} \times \abs{\mathbf{s}^{d}}$ matrix by defining $\Gamma^{d}_{ax}(s, t) := \phi^{d}_{ax}(s^{\dagger}t)$. $\Gamma^{d}_{ax}$ is positive semidefinite and satisfies $\Gamma^{d}_{ax}(s_{1}, t_{1}) = \Gamma^{d}_{ax}(s_{2}, t_{2})$ whenever $s_{1}^{\dagger}t_{1} = s_{2}^{\dagger}t_{2}$ in $\mathcal{A}_{PVM}^{\mathcal{B}, \mathcal{Y}}$. Let $\mathbf{C}^{ax}_{k}$ be the matrix encoding these constraints as in \Cref{eq:npa_formulation} for each $a \in \mathcal{A}, x \in \mathcal{X}$. Then, a positive semidefinite matrix $\Gamma^{d}_{ax}$ satisfying $\innerprod{\Gamma^{d}_{ax}}{\mathbf{C}^{ax}_{k}} = 0$ for all $k$ induces a linear functional $\phi^{d}_{ax}$ by defining $\phi^{d}_{ax}(t) = \Gamma_{ax}^{d}(I,t)$. So then, let $\Gamma^{d} := \bigoplus_{a,x} \Gamma^{d}_{ax}$ and $\mathbf{P}_{\mathcal{G}}$ be a $\abs{\mathbf{s}^{d}}^{\abs{\mathcal{A}}\abs{\mathcal{X}}} \times \abs{\mathbf{s}^{d}}^{\abs{\mathcal{A}}\abs{\mathcal{X}}}$ game polynomial matrix such that 
\be \innerprod{\mathbf{P}_{\mathcal{G}}}{\Gamma^{d}} = \sum_{a,b,x,y} c_{abxy} \Gamma^{d}_{ax}(1,N_{by}). \ee 
Furthermore, let $\mathbf{B}_{x,s,t} := \sum_{a \in \mathcal{A}} \II^{a0}_{s, t} - \sum_{a \in \mathcal{A}} \II^{ax}_{s, t}$, where $\II^{ax}_{s,t}$ is the matrix which is $1$ in entry $(s,t)$ in the $ax$ block and $0$ everywhere else, which encodes the consistency constraint.

So, replacing all instances of $\phi_{ax}^{d}$ with $\Gamma^{d}_{ax}$ evaluations above, we obtain:

\begin{equation} \label{one_sided_npa_primal}
    \begin{alignedat}{3}
        \mathbbm{p}^{d}(\mathcal{G}) \quad = \qquad\qquad & \max_{\mathclap{ \cbrac{ \Gamma^{d}_{ax} }_{a,x} } } \qquad\qquad && \innerprod{\mathbf{P}_{\mathcal{G}}}{\Gamma^{d}} && \\
        & \,\,\mathrm{s.t.} \quad && \inprod{ \mathbf{B}_{x,s,t} }{ \Gamma^{d} } = 0, \quad \forall x \in \mathcal{X}, s,t \in \mathbf{s}^{d} && \quad \text{(consistency)}\\
        & && \innerprod{\sum_{a \in \mathcal{A}} \II^{a0}_{I, I} }{\Gamma^{d}} = 1, &&\quad \text{(identity)} \\
        & && \innerprod{ \mathbf{C}^{ax}_{k} }{\Gamma^{d}} = 0, \quad \forall a,x,k &&\quad \\
        & && \,\, \Gamma^{d} \succeq 0, &&
    \end{alignedat}
\end{equation}
where in the last constraint, we are using the fact that $\Gamma^{d} = \bigoplus_{a,x} \Gamma_{ax}^{d}$ is positive semidefinite if and only if $\Gamma_{ax}^{d}$ is positive semidefinite for all $a, x$. Finally, to be strictly in ``standard form'' for SDP, we should be maximizing over $\Gamma^{d}$ with no matrix structure instead of over $\cbrac{ \Gamma^{d}_{ax} }_{a,x}$, but this is easy to enforce by just adding additional constraints to make $\Gamma^{d}$ block diagonal of the form $\Gamma^{d} = \bigoplus_{a,x} \Gamma_{ax}^{d}$.

\paragraph{Dual of one-sided NPA Hierarchy---the nice SoS hierarchy}
Given the above standard form, we can immediately write the dual of the one-sided NPA hierarchy as follows:
\begin{equation} \label{eq:npa_lin_formulation}
    \begin{alignedat}{2}
        \mathbbm{d}^{d}(\mathcal{G}) \quad = \qquad & \min_{\mathclap{ \nu, y_{x, s, t}, y_{a,x,k} } } \qquad && \nu \\
        & \,\,\text{s.t.} \quad && \mathbf{M}_{d} := \nu \II_{1, 1}^{a0} + \sum_{x, s, t} y_{x, s, t} \mat{B}_{x, s, t} + \sum_{a,x, k} y_{a,x,k} \mat{C}^{ax}_k - \mat{P}_{\mathcal{G}}\succeq 0.
    \end{alignedat}
\end{equation}

Note that these constraint matrices are block-diagonal, where each block is limited to only one question of Alice. Thus, the $\mat{M}_d$ matrix can be decomposed into $\mat{R}^\dagger \mat{R}$ where $\mat{R}$ represents the nice sum-of-squares decomposition of polynomials up to degree $d$. So, this hierarchy finds the best upper-bound by searching over nice SoS certificates of degree $d$.

We will talk more about the matrix structure of nice SoS certificates in \Cref{chap:npa_level1}. 

\subsection{Convergence of the one-sided NPA hierarchy} \label{sec:conv_oneNPA_hier}

We will now show that the one-sided NPA hierarchy converges to the optimal commuting value of the game, which would give us a new hierarchy of upper bounds on the value of nonlocal games.

\begin{thm} \label{thm:one-sided-convergence}
    For any nonlocal game $\mathcal{G}$, let $\mathbbm{p}^d(\mathcal{G})$ be the optimal value of the $d$-th level of the one-sided NPA hierarchy. This optimal value converges to the commuting value of the game, i.e.,
    \begin{equation}
        \lim_{d \to \infty} \mathbbm{p}^d(\mathcal{G}) = \omega_{qc}^{\ast}(\mathcal{G}).
    \end{equation}
\end{thm}
The proof strategy is similar to one used in Theorem 6.1 of \cite{2408.06711} and Theorem 8 of \cite{0803.4290}. 

\begin{proof}
    Firstly, note that $\mathbbm{p}^d(\mathcal{G})$ is monotonically non-increasing in $d$. Hence, the limit $\lim_{d\to\infty} \mathbbm{p}^d(\mathcal{G})$ exists. Also, note that $\mathbbm{p}^d(\mathcal{G}) \geq \omega_{qc}^{\ast}(\mathcal{G})$ for all $d$, as the optimal commuting operator strategy is a valid strongly non-signaling algebraic strategy for $d$-degree moments. Hence, 
    \be \label{eq:convergence=proof-geq}
    \lim_{d\to\infty} \mathbbm{p}^d(\mathcal{G}) \geq \omega_{qc}^{\ast}(\mathcal{G}). \ee

    Now, we will show that $\lim_{d\to\infty} \mathbbm{p}^d(\mathcal{G}) \leq \omega_{sns}^{\ast}(\mathcal{G})$ by constructing a non-signaling algebraic strategy from the sequence of optimal moment matrices $\Gamma^d$.\\
    We can extend $\Gamma^d$ linearly to get the linear functionals $\phi^d_{ax} : (\mathcal{A}_{PVM}^{\mathcal{A}, \mathcal{X}})_{\leq 2d} \to \mathbb{C}$. These functionals are also positive on this space. Furthermore, we can extend it to the whole of $\mathcal{A}_{PVM}^{\mathcal{A}, \mathcal{X}}$ by defining $\phi^d_{ax} = 0$ on monomials of degree strictly greater than $2d$. This extended version of $\phi^d_{ax}$ is still positive semidefinite.\\
    Note that each $x, a$ and for all $d \geq 1$, the operator norm of $\phi^d_{ax}$ is bounded by 1, i.e.,
    \be \norm{\phi^d_{ax}} = \sup_{v: \norm{v} \leq 1} \phi^d_{ax}(v) = \phi^d_{ax}(I) \leq \sum_{a} \phi^d_{ax}(I) = 1.\ee
    Hence, the sequence of positive linear functionals $\phi^d_{ax}$ is bounded in operator norm. 
    By the Banach-Alaoglu theorem, the sequence $\cbrac{\phi^d_{ax}}_{d\in\NN}$ (and any of its subsequence) has a weak-$*$ convergent subsequence. As $\alpbt{X}$ and $\alpbt{A}$ are finite sets, by iteratively taking subsequences for each $x$ and $a$, we can find an increasing subsequence $\cbrac{d_k}_{k\in\NN}$ and a positive linear functional $\phi_{ax}$ such that
    \be \lim_{k\to\infty} \phi^{d_k}_{ax} = \phi_{ax} \quad \forall x \in \alpbt{X}, a \in \alpbt{A}. \ee
    The limit functional $\phi_{ax}$ is also positive as it is the limit of positive functionals. Furthermore, for any $x, x' \in \alpbt{X}$,
    \be \sum_{a\in\alpbt{A}} \phi_{a|x} = \sum_{a\in\alpbt{A}} \lim_{k\to\infty} \phi^{d_k}_{ax} = \lim_{k\to\infty} \sum_{a\in\alpbt{A}} \phi^{d_k}_{ax} = \lim_{k\to\infty} \sum_{a\in\alpbt{A}} \phi^{d_k}_{ax'} = \sum_{a\in\alpbt{A}} \phi_{ax'}. \ee
    Hence, we can define 
    \be \phi = \sum_{a\in\alpbt{A}} \phi_{ax}, \ee
    which is positive linear functional on $\mathcal{A}_{PVM}^{\mathcal{A}, \mathcal{X}}$. This gives rise to a non-signaling algebraic strategy.
    Note that,
    \be \phi(I) = \sum_{a\in\alpbt{A}} \phi_{ax}(I) = \lim_{k\to\infty} \sum_{a\in\alpbt{A}} \phi^{d_k}_{ax}(I) = 1, \ee
    so $\phi$ is a valid state.
    
    Hence, $\phi_{ax}: \mathcal{A}_{PVM}^{\mathcal{A}, \mathcal{X}} \to \CC$ forms a non-signaling algebraic strategy with PVM operators $\cbrac{N_{by}}_{y\in\alpbt{Y}}$. So, 
    \be 
    \lim_{d\to\infty} \mathbbm{p}^d(\mathcal{G}) = \lim_{k\to\infty} \mathbbm{p}^{d_{k}}(\mathcal{G}) \leq \omega_{sns}^{\ast}(\mathcal{G}). \ee
    From lemma \ref{lem:non-signaling-commuting}, we have $\omega_{sns}^{\ast}(\mathcal{G}) \leq \omega_{qc}^{\ast}(\mathcal{G})$. Hence, 
    \be \label{eq:convergence=proof-leq}
    \lim_{d\to\infty} \mathbbm{p}^d(\mathcal{G}) \leq \omega_{qc}^{\ast}(\mathcal{G}). 
    \ee
    Combining the two inequalities \eqref{eq:convergence=proof-geq} and \eqref{eq:convergence=proof-leq}, we get
    \be \lim_{d\to\infty} \mathbbm{p}^d(\mathcal{G}) = \omega_{qc}^{\ast}(\mathcal{G}). \ee
\end{proof}

So, we have constructed a convergent SDP hierarchy for the commuting operator value of the game. We can combine this result with duality of the one-sided NPA hierarchy and the nice SoS hierarchy and apply Theorem~\ref{thm:nice-compiles} to get the following main result of the paper: 

\begin{thm} \label{thm:formal-compilation-bound}
Let $\game{G}$ be a nonlocal game with optimal quantum value $\omega_{qc}^{\ast}(\game{G})$. For any $\eps > 0$, there exists $d(\eps) \in \NN$ such that there is a nice sum-of-squares certificate of degree $d(\eps)$ certifying that the optimal quantum value of the game is at most $\omega_{qc}^{\ast}(\game{G}) + \eps$. This implies that any computationally bounded prover strategy $\strat{S}$ on the compiled game can win $\game{G}$ with value at most 
\be \compgval{S}{G} \leq \omega_{qc}^{\ast}(\game{G}) + \eps + \negl_{\strat{S}, d(\eps)}(\lambda), \ee where $\negl_{\strat{S}, d(\eps)}$ is a negligible function of $\lambda$ that depends on the strategy $\strat{S}$ and $d(\eps)$. 
\end{thm}

We note that additionally if at any finite level $d$, the value of the SoS hierarchy $\mathbbm{d}^{d}(\mathcal{G}) = \omega_{qc}^{\ast}(\mathcal{G})$, then the game $\mathcal{G}$ compiles in the sense that
\[ \compgval{S}{G} \leq \omega_{qc}^{\ast}(\game{G}) + \negl_{\strat{S}}(\lambda). \]

\begin{proof}[Proof of \Cref{thm:formal-compilation-bound}]
    Let $\mathbbm{p}^d(\mathcal{G})$ be the optimal value of the $d$-th level of the one-sided NPA hierarchy. From Theorem \ref{thm:one-sided-convergence}, we know that 
    \be \lim_{d\to\infty} \mathbbm{p}^d(\mathcal{G}) = \omega_{qc}^{\ast}(\game{G}). \ee
    This means that for any $\eps > 0$, there exists a $d(\eps) \in \NN$ such that
    \be \mathbbm{p}^{d(\eps)}(\mathcal{G}) \leq \omega_{qc}^{\ast}(\game{G}) + \eps. \ee
    From the optimal solution of the dual of the $d(\eps)$-level of one-sided NPA hierarchy, we know that there exists a nice SoS certificate of degree $d(\eps)$ certifying that the value of the game is at most $\mathbbm{p}^{d(\eps)}(\mathcal{G})$. Hence, we have a nice SoS certificate of degree $d(\eps)$ certifying that the value of the game is at most $\omega_{qc}^{\ast}(\game{G}) + \eps$.\\
    Then, we can apply Theorem \ref{thm:nice-compiles} to get that any computationally bounded prover strategy $\strat{S}$ can win the compiled game with probability at most
    \be \compgval{S}{G} \leq \omega_{qc}^{\ast}(\game{G}) + \eps + \negl_{\strat{S}, d(\eps)}(\lambda). \ee

\end{proof}

The above result implies the Theorem 6.1 of \cite{2408.06711}.
\begin{cor}[\cite{2408.06711}] \label{cor:2408.generalization}
    Let $\game{G}$ be a nonlocal game and $\strat{S}$ be a computationally bounded quantum prover for the compiled game. Then,
    \be \limsup_{\lambda \to \infty} \compgval{S}{G} \leq \omega_{qc}^{\ast}(\game{G}). \ee
\end{cor}
\begin{proof}
    From Theorem \ref{thm:formal-compilation-bound}, we know that for any $\eps > 0$,
    \be \compgval{S}{G} \leq \omega_{qc}^{\ast}(\game{G}) + \eps + \negl_{\strat{S}, d(\eps)}(\lambda). \ee
    Taking the limit $\lambda \to \infty$, we get that for any $\eps > 0$,
    \bearr \limsup_{\lambda \to \infty} \compgval{S}{G} \leq \omega_{qc}^{\ast}(\game{G}) + \eps. \eearr
    Assume the contrary that $\limsup_{\lambda \to \infty} \compgval{S}{G} > \omega_{qc}^{\ast}(\game{G})$. Then, there exists an $\eps > 0$ such that
    \be \limsup_{\lambda \to \infty} \compgval{S}{G} - \omega_{qc}^{\ast}(\game{G}) > \eps, \ee
    which contradicts the above inequality. Hence, we must have
    \be \limsup_{\lambda \to \infty} \compgval{S}{G} \leq \omega_{qc}^{\ast}(\game{G}). \ee
\end{proof}

\section{A computational Tsirelson's theorem for all NPA level-1 games}\label{chap:npa_level1}

In this section, we will show that any NPA level-1 SoS decomposition of a game polynomial of a nonlocal game can be re-expressed as a {\nice} SoS decomposition at level 1. This shows that compilation preserves the quantum soundness at NPA level 1. In this section, we will refer to NPA level-1 games. These are nonlocal games for which level 1 of the NPA hierarchy is sufficient to bound the quantum value. 

We shall give two proofs of this. The first is a ``matrix-theoretic'' proof, which exploits the freedom in the Cholesky decomposition of a positive semidefinite matrix. The second gives a more ``NPA-theoretic'' proof that transforms the Gram vectors of the moment matrices, this is reminiscent of Tsirelson-type proofs where one manipulates Gram vectors to construct feasible correlations.

\subsection{The Cholesky decomposition approach}

We begin with a technical lemma about the Cholesky decomposition. 

\subsubsection{Choice in Cholesky decomposition principal submatrices} \label{sec:cholesky-freedom}
We know that the Cholesky decomposition of a positive semidefinite matrix is not unique. In this section, we will show that we can choose an arbitrary factorization for the top left block of a matrix~$M$ by adapting the rest of the decomposition to a valid Cholesky decomposition. Specifically, we show the following lemma:

\begin{lem} \label{lem:cholesky-freedom}
    Let $M \in \mathbb{C}^{n \times n}$ be a positive semidefinite matrix, with the following block structure:
    \be    
        M = \blockmat{M_a}{M_{ab}}{M^\dagger_{ab}}{M_b},
    \ee
    where $M_a, M_{ab}, M_b$ are block matrices, and $R_a$ give a specific Cholesky decomposition of $M_a$, i.e., $M_a = R^\dagger_a R_a$. Then, we can ``complete'' the Cholesky decomposition given by $R_{a}$ with
    \be
        R = \blockmat{R_a}{R_{ab}}{0}{R_b}
    \ee
    such that $M = R^{\dagger}R$, for some block matrix $R_{ab}$ and upper triangular matrix $R_b$.
\end{lem}

\begin{proof}
    Firstly, note that $M_a$ is necessarily positive-definite as all principal submatrices of a positive semidefinite matrix are positive semidefinite. Take an arbitrary Cholesky decomposition of $M = S^\dagger S$ such that
    \bea
    S = \blockmat{S_a}{S_{ab}}{0}{S_b},
    \eea
    then from the Cholesky decomposition
    \bea
    \blockmat{M_a}{M_{ab}}{M^\dagger_{ab}}{M_b} = \blockmat{S^\dagger_a}{0}{S^\dagger_{ab}}{S^\dagger_b} \blockmat{S_a}{S_{ab}}{0}{S_b},
    \eea
    we get the following relations from multiplying the matrix blocks:
    \bea
    M_a &=& S^\dagger_a S_a, \\
    M_{ab} &=& S^\dagger_a S_{ab}, \\
    M_b &=& S^\dagger_b S_b.
    \eea

    Hence, $S_a$ is a valid Cholesky decomposition of $M_a$. Note that $R_a$ is also a valid Cholesky decomposition of $M_a$. Specifically, both $R_a$ and $S_a$ are Gram decomposition of the matrix $M_a$. So, we can find a unitary matrix $V$ such that $R_a = V S_a$. This is a well-known lemma that follows from Theorem 7.3.11 of \cite{Horn_Johnson_2013}. We can now construct the needed Cholesky decomposition of $M$ as follows. Let
    \be
    R = \blockmat{R_a}{R_{ab}}{0}{R_b} = \blockmat{V S_a}{V S_{ab}}{0}{S_b},
    \ee
    and verify that $R^\dagger R = M$ by computing
    \bea
    R^\dagger R &=& \blockmat{S^\dagger_a V^\dagger}{0}{S^\dagger_{ab} V^\dagger}{S^\dagger_b} \blockmat{V S_a}{V S_{ab}}{0}{S_b}\\
    &=& \blockmat{S^\dagger_a V^\dagger V S_a}{S^\dagger_a V^\dagger V S_{ab}}{S^\dagger_{ab} V^\dagger V S_a}{S^\dagger_{ab} V^\dagger V S_{ab} + S^\dagger_b S_b} \\
    &=& \blockmat{S^\dagger_a S_a}{S^\dagger_a S_{ab}}{S^\dagger_{ab} S_a}{S^\dagger_{ab} S_{ab} + S^\dagger_b S_b}\\
    &=& \blockmat{M_a}{M_{ab}}{M^\dagger_{ab}}{M_b} \\
    &=& M.
    \eea
    Hence, $R$ gives a valid Cholesky decomposition of $M$, as desired.
\end{proof}

\subsubsection{Unitary freedom of SoS decompositions} \label{subsec: unitary freedom}
Here, we establish that given a particular SoS decomposition, we can \qq{apply a unitary} to the coefficients of the SoS to obtain another SoS decomposition. More specifically, note that any SoS decomposition of some polynomial $Q \in \mathcal{A}_{\mathcal{G}}$ can be expressed as
\begin{eqnarray} \label{eq:sos_matrix_form}
    Q = \vec{b}^\dagger S^\dagger S \vec{b},
\end{eqnarray}
where $\vec{b}$ is the basis vector of monomials of operators and $S$ is the matrix of coefficients. Each row in the matrix $S$ corresponds to a polynomial term in the SoS decomposition.

Now, note that, we can modify the matrix $S$ by applying any unitary $U$ as follows without changing the SoS decomposition. Indeed,
\begin{eqnarray}
    \vec{b}^\dagger (US)^\dagger (US) \vec{b} = \vec{b}^\dagger S^\dagger U^\dagger U S \vec{b} = \vec{b}^\dagger S^\dagger S \vec{b} = Q.
\end{eqnarray}
Thus, we can apply a unitary to the matrix $S$ while still remaining a SoS decomposition for $Q$.

Now, any matrix $S$ has a QR decomposition, $S = QR$, where $Q$ is unitary and $R$ is upper triangular. This, together with the above observation, implies that any SoS decomposition
\[ Q = \vec{b}^{\dagger} S^{\dagger}S \vec{b} \]
has an ``upper triangular'' SoS decomposition
\[ Q = \vec{b}^{\dagger} R^{\dagger}R \vec{b}. \]

\subsubsection{Structure of {\nice} sum-of-squares decomposition} \label{sec:nice-structure}

From the discussion in the previous subsection, we know that we can always obtain an upper triangular SoS decomposition. Recall that we said in \Cref{def:niceness} that a SoS decomposition is nice if every square term in the decomposition contains powers of at most one Alice operator. In other words, it contains monomials corresponding to only one question of Alice. So, each row in $S$, has non-zero entries corresponding to at most one Alice's operator.

Suppose we have a SoS decomposition from level-1 of the SoS hierarchy. The basis will be, after a transformation into observables as described in \Cref{sec:observables},
\bearr
    \vec{b} = \{ M_{1}, \dots, M_{1}^{d-1}, \dots, M_{k}, \dots, M_{k}^{d-1}, \dots, N_{1}, \dots, N_{1}^{d-1}, \dots, N_{k}^{d-1}, \dots, N_{k}^{d-1}, 1\},
\eearr
where $k$ is the number of questions and $d$ is the number of answers. Hence, the length of this basis is $2k(d-1) + 1$. Then, the {\nice} SoS decomposition will take the following form:

\begin{equation} \label{eqn:nice_SOS_matrix}
\left(
\begin{array}{c c c c||c}
M_1 & \mathbf{0} & \cdots & \mathbf{0} & \\
\cline{1-4}
\mathbf{0} & M_2 & \cdots & \mathbf{0} & \multirow{5}{*}{{\Huge $N$}} \\
\cline{1-4}
\vdots & \vdots & \ddots & \vdots & \\
\cline{1-4}
\mathbf{0} & \mathbf{0} & \cdots & M_k & \\
\cline{1-4}
\multicolumn{4}{c||}{\rule{0pt}{4ex}\text{\Large $\mathbf{0}$}\rule[-2ex]{0pt}{0pt}} & 
\end{array}
\right).
\end{equation}

Here, each matrix $M_i$ are $(d-1) \times (d-1)$ block upper-triangular matrices and $N$ is a $k(d-1) + 1$ wide block matrix.

The goal of the next section shall be to show that any SoS decomposition from NPA level 1 can be transformed to this nice block form.

\subsubsection{Compilation preserves NPA level-1 value} \label{sec:lvl1_npa_lvl1_nice}

Let $\mathcal{G}$ be a nonlocal game with $k$ questions and $d+1$ answers. As described in \Cref{sec:observables}, the game polynomial can be represented as a sum of monomials of the form $M_x^j N_y^k$.
\begin{eqnarray}
    P_{\mathcal{G}} = \sum_{j,k, x,y} d_{j,k,x,y} M_{x}^{j} N_{y}^{k}.
\end{eqnarray}
Suppose we have a SoS decomposition for $\mathfrak{p}^{1}(\mathcal{G}) I - P_{\mathcal{G}}$. From \Cref{eq:sos_matrix_form}, we know that this can be expressed as
\[ \mathfrak{p}^{1}(\mathcal{G}) I - P_{\mathcal{G}} = \mathbf{b}^{\dagger}S^{\dagger}S\mathbf{b}, \]
for some coefficient matrix $S$. Let $M := S^{\dagger}S$.

Now, we know that $M$ must have the form: 
\begin{equation} \label{eqn:M_nice_structure}
    M = \begin{pNiceArray}[first-col, first-row]{cccc|ccccc}
        & M_1, \cdots M_1^d & M_2, \cdots M_2^d & \cdots & M_k, \cdots M_k^d & N_1, \cdots N_1^d &  & \cdots & N_k, \cdots N_k^d & I \\
    M_1, \cdots, M_1^d & \mat{M}_1 & \mathbf{0} & \cdots & \mathbf{0} & \Block{4-5}{\text{\Large$\mat{C}_{ab}$}} & &  &  & \\
    M_2, \cdots, M_2^d & \mathbf{0} & \mat{M}_2 & \cdots & \mathbf{0} &  & &  &  & \\
    \vdots & \vdots & \vdots & \ddots & \vdots &  & &  & &  \\
    M_k, \cdots, M_k^d & \mathbf{0} & \mathbf{0} & \cdots & \mat{M}_k &  & &  &  & \\
    \hline
    N_1, \cdots, N_1^d & \Block{5-4}{\text{\Large$\mat{C}^\dagger_{ab}$}} &  &  &  & \Block{5-5}{\text{\Large$\mat{M}_{b}$}} &  &  & & \\
    &  &  &  &  & &  &  & & \\
    \vdots &  &  &  &  & &  &  & & \\
    N_k, \cdots, N_k^d &  &  &  &  & &  &  & & \\
    I &  &  &  &  & &  &  & &
\end{pNiceArray}
\end{equation}

Here, $\mat{M}_i$ are $d \times d$ block matrices, $\mat{M}_{b}$ is a $k(d+1) \times k(d+1)$ matrix, and $\mat{C}_{ab}$ is a $kd \times k(d+1)$ matrix. First, we show the following lemma which will help us prove that $M$ has a {\nice} SoS decomposition.
\begin{lem}\label{lem:block_cholesky}
    A block-diagonal matrix has a block-diagonal Cholesky decomposition.
\end{lem}
\begin{proof}
    Firstly, note that if a block-diagonal matrix is positive semidefinite, then each block is positive semidefinite as all principal submatrices of a positive semidefinite matrix are positive semidefinite. 
    
    So, we can decompose each block $M_i$ as $M_i = R_i^\dagger R_i$. Now, we can construct a block-diagonal matrix $R$ as $R = \begin{pmatrix} R_1 & 0 & \cdots & 0 \\ 0 & R_2 & \cdots & 0 \\ \vdots & \vdots & \ddots & \vdots \\ 0 & 0 & \cdots & R_k \end{pmatrix}$.\\
    Now, 
    \begin{equation}
        R^\dagger R = \begin{pmatrix} R_1^\dagger R_1 & 0 & \cdots & 0 \\ 0 & R_2^\dagger R_2 & \cdots & 0 \\ \vdots & \vdots & \ddots & \vdots \\ 0 & 0 & \cdots & R_k^\dagger R_k \end{pmatrix} = \begin{pmatrix}
        M_1 & 0 & \cdots & 0 \\
        0 & M_2 & \cdots & 0 \\
        \vdots & \vdots & \ddots & \vdots \\
        0 & 0 & \cdots & M_k
    \end{pmatrix} = M.
\end{equation}
    Hence, we have shown that a block-diagonal matrix has a block-diagonal Cholesky decomposition.
\end{proof}

\begin{thm} \label{thm:lvl1-nicelvl1}
    Let $\game{G}$ be a nonlocal game with game polynomial $\gp{G}$ in Alice and Bob PVMs $\Apovms$ and $\Bpovms$ respectively. If we have a degree-$1$ NPA sum-of-squares certificate for $\mathfrak{p}^{1}(\mathcal{G}) I - \gp{G}$, i.e. 
    \be
        \mathfrak{p}^{1}(\mathcal{G}) I - \gp{G} = \sum_{i} r_i^\dagger r_i ,
    \ee
    where $r_i$ are linear polynomials in the PVMs. Then, we can construct a degree-$1$ {\nice} sum-of-squares certificate for $\mathfrak{p}^{1}(\mathcal{G}) I - \gp{G}$ of the form 
    \be
        \mathfrak{p}^{1}(\mathcal{G}) I - \gp{G} = \sum_{i} {r'_i}^\dagger r'_i,
    \ee
    where $r'_i$ are linear polynomials in the PVMs and each $r'_i$ contains terms corresponding to only one question of Alice.
\end{thm}

\begin{proof}
    Construct a basis of linear monomials of the PVMs as $\vec{b} = \cbrac{\Apovms, \Bpovms, \II}$. Choose an arbitrary answer for Alice and Bob, say $0$ and $0$. Remove all the POVMs $\cbrac{M_{0x}}_{x\in \alpbt{X}}, \cbrac{N_{0y}}_{y\in \alpbt{Y}}$ from the basis so that the basis elements are linearly independent.

    Now, each $r_i$ can be expressed as a linear combination of the remaining monomials. Let $\vec{r}_i$ be the vector of coefficients of $r_i$ in this basis such that $r_i = \vec{r}_i \vec{b}$. Construct the matrix $R$ with rows as $\vec{r}_i$, such that 
    \be \sum_i r_i^\dagger r_i = \vec{b}^\dagger R^\dagger R \vec{b}. \ee
    Let $M = R^\dagger R$. We know that $M$ should have a nice structure as shown in \Cref{eqn:M_nice_structure}. So, $M$ has the structure,
\be M = \blockmat{\mat{M_a}}{\mat{C_{ab}}}{\mat{C_{ab}}^\dagger}{\mat{M_b}}, \ee
where $\mat{M_a}$ is a block-diagonal matrix. From \Cref{lem:block_cholesky}, we know that $\mat{M_a}$ has a block-diagonal Cholesky decomposition. Applying \Cref{lem:cholesky-freedom}, we can obtain a Cholesky decomposition of $M = (R^{\prime})^\dagger R^{\prime}$ of the form:
\be
    R' = \blockmat{\mat{R}_a}{\mat{R}_{ab}}{0}{\mat{R}_b},
\ee
where $\mat{R}_a$ is block-diagonal. Each block of $\mat{R}_a$ corresponds to one question of Alice. So, $M$ has a {\nice} SoS decomposition given by $R'$. From this {\nice} decomposition, we can construct the corresponding {\nice} SoS decomposition $r'_i$, such that 
\be
    \sum_i {r'_i}^\dagger r'_i = \sum_i r_i^\dagger r_i.
\ee
This gives the desired {\nice} SoS decomposition
\be
        \omega - \gp{G} = \sum_{i} {r'_i}^\dagger r'_i,
    \ee
as desired.
\end{proof}

\begin{cor}
    Let $\game{G}$ be a nonlocal game where the quantum commuting value is achieved at NPA level 1. For any computationally bounded strategy $\mathcal{S}$, the winning probability of strategy is bounded as 
    \be \compgval{S}{G} \leq \omega_{qc}(\game{G}) + \negl_{\mathcal{S}}(\lambda) \ee
    where $\omega_{qc}(\game{G})$ is the quantum commuting value of the game and $\negl_{\mathcal{S}}(\lambda)$ is a negligible function of the security parameter $\lambda$.
\end{cor}
\begin{proof}
As the quantum value of $\mathcal{G}$ is achieved at level 1, we know that there is a degree-1 SoS certificate for $\omega_{qc}(\mathcal{G}) - \gp{G}$. From \Cref{thm:lvl1-nicelvl1}, we know that we can construct a {\nice} SoS decomposition for $\omega_{qc}(\mathcal{G}) - \gp{G}$. 
Applying \Cref{thm:nice-compiles} gives the desired result.
\end{proof}

This result encapsulates the result for XOR games given in \cite{Cui+24}.

\subsection{The Gram vector approach}

In this section, we give an alternate proof for \Cref{thm:lvl1-nicelvl1}. To show this, we just need to show that the NPA level 1 value is equal to the one-sided NPA level 1 value. In particular, we show that every feasible solution for NPA level 1 has a matching feasible solution for one-sided NPA level 1 with the same value and vice-versa. 

We begin by stating a standard fact of linear algebra.

\begin{lem} \label{lem:unitary_lem}
    For two sets $E = \cbrac{v_{i}}_{i=1}^{k}$, $F = \cbrac{w_{i}}_{i=1}^{k}$ of vectors, there exists a unitary mapping $E$ to $F$ if and only if $\inprod{v_{i}}{v_{j}} = \inprod{w_{i}}{w_{j}}$ for all $i, j \in [k]$.
\end{lem}

\begin{proof}[Proof of \Cref{thm:lvl1-nicelvl1}]

    Let $\Gamma$ be a feasible solution to the primal NPA hierarchy as given in \Cref{eq:npa_formulation}. This restricts to a feasible solution for \Cref{one_sided_npa_primal} with the same value. 

    Now, let us fix a level-1 feasible solution $\Gamma$ of \Cref{one_sided_npa_primal} with blocks given by $\Gamma^{ax}$. Since $\Gamma^{ax} \geq 0$ for all $a \in \mathcal{A}, x \in \mathcal{X}$ we have Gram vectors $\cbrac{ w_{I^{ax}} } \cup \cbrac{ w_{N_{by}^{ax}} }_{b \in \mathcal{B},y \in \mathcal{Y}}$, where the vector indexed by $I^{ax}$ corresponds to the Gram vector for the $ax$ block for the $I$ row/column and $N^{ax}_{by}$ corresponds to the $ax$ block for the $N_{by}$ row/column. Note that
    \begin{align*}
        \inprod{ \sum_{a} w_{N_{by}^{ax}} \otimes \ket{a} }{ \sum_{a^{\prime}} w_{ N_{b^{\prime}y^{\prime}}^{a^{\prime}x} } \otimes \ket{a^{\prime}} } &= \sum_{a} \inprod{ w_{N_{by}^{ax}} }{ w_{N_{b^{\prime}y^{\prime}}^{ax}} } \\
            &= \sum_{a} \Gamma^{ax}(N_{by}, N_{b^{\prime}y^{\prime}}) \\
            &= \sum_{a} \Gamma^{a0}(N_{by}, N_{b^{\prime}y^{\prime}}) \\
            &= \inprod{ \sum_{a} w_{N_{by}^{a0}} \otimes \ket{a} }{ \sum_{a^{\prime}} w_{ N_{b^{\prime}y^{\prime}}^{a^{\prime}0} } \otimes \ket{a^{\prime}} }.
    \end{align*}
    For each $x \in \mathcal{X}$, set 
    \[ E_{x} := \cbrac{ \sum_{a} w_{I^{ax}} \otimes \ket{a} } \cup \cbrac{ \sum_{a} w_{ B_{by}^{ax} } \otimes \ket{a} }. \]
    Then by \Cref{lem:unitary_lem}, there exists a unitary $U_{x}$ sending $E_{x}$ to $E_{0}$ for each $x \in \mathcal{X}$. 

    We shall now define Gram vectors for the standard NPA hierarchy. Let
    \begin{align*}
        v_{M_{ax}} &:= U_{x} \brac{ w_{I^{ax}} \otimes \ket{a} }, \\
        v_{N_{by}} &:= \sum_{a} w_{ N_{by}^{a0} } \otimes \ket{a}, \\
        v_{I} &:= \sum_{a} w_{I^{a0}} \otimes \ket{a}.
    \end{align*}

    The Gram matrix $\Lambda$ of this set of vectors will be our feasible solution to \Cref{eq:npa_formulation}. Clearly $\Lambda$ is positive semidefinite (as it is a Gram matrix) and so we just need to check that $\Lambda$ satisfies all the constraints and its objective value matches that of $\Gamma$. We check the latter condition first:
    \begin{align*}
        \Lambda(M_{ax}, N_{by}) &= \inprod{ v_{M_{ax}} }{ v_{N_{by}} } \\
                          &= \inprod{ U_{x}\brac{ w_{I^{ax}} \otimes \ket{a} } }{ \sum_{a^{\prime}} w_{ N_{by}^{a^{\prime}0} } \otimes \ket{a^{\prime}} } \\
                          &= \inprod{  w_{I^{ax}} \otimes \ket{a} }{ U_{x}^{\ast} \brac{ \sum_{a^{\prime}} w_{ N_{by}^{a^{\prime}0} } \otimes \ket{a^{\prime}} } } \\
                          &= \inprod{  w_{I^{ax}} \otimes \ket{a} }{ \sum_{a^{\prime}} w_{ N_{by}^{a^{\prime}x} } \otimes \ket{a^{\prime}} } \\
                          &= \inprod{ w_{I^{ax}} }{ w_{ N_{by}^{ax} } } \\
                          &= \Gamma^{ax}(I, N_{by}),
    \end{align*}
    and hence $\sum_{a,b,x,y} c_{a,b,x,y} \Lambda(M_{ax}, N_{by}) = \sum_{a,b,x,y} c_{a,b,x,y} \Gamma^{ax}(I, N_{by})$ which is exactly the equality of optimization values for these two hierarchies.

    Now, to check the identity constraint
    \[ \Lambda(I, I) = \inprod{ v_{I} }{ v_{I} } = \sum_{a} \inprod{ w_{I^{a0}} }{ w_{I^{a0}} } = \sum_{a} \Gamma^{a0}(I, I) = 1. \]

    Next, we check the algebraic constraints imposed by the game algebra $\mathcal{A}_{\mathcal{G}}$. Since we are in level $1$, the only algebraic relations to check are $\sum_{a} M_{ax} = \sum_{b} N_{by} = I$ and that $M_{ax}N_{by} = N_{by}M_{ax}$.
    
    The commutativity constraint is easy to check as
    \[ \Lambda(M_{ax}, N_{by}) = \Gamma^{ax}(I, N_{by}) = \Gamma^{ax}(N_{by}, I) = \Lambda(N_{by}, M_{ax}),  \]
    where the first and last equality is from the computation of $\Gamma(M_{ax}, N_{by})$ above, and the second equality is from $\Gamma^{ax}(s_{1}, t_{1}) = \Gamma^{ax}(s_{2}, t_{2})$ whenever $s_{1}^{\dagger}t_{1} = s_{2}^{\dagger}t_{2}$.
    
    Now, we move on to check that $\sum_{b} N_{by} = I$. To check this, we need to check that $\Lambda(Q, \sum_{b} N_{by}) = \Lambda(Q, I)$ for any $Q \in \mathbf{b}^{1}$. We first check $\sum_{b} N_{by}$ against the identity
    \begin{align*}
        \inprod{ v_{I} }{ \sum_{b} v_{N_{by}} } &= \inprod{ \sum_{a} w_{I^{a0}} \otimes \ket{a} }{ \sum_{b} \sum_{a^{\prime}} w_{ N^{a^{\prime}0}_{by} } \otimes \ket{a^{\prime}} } \\
            &= \sum_{a} \inprod{ w_{I^{a0}} }{ \sum_{b} w_{ N_{by}^{a0} } } \\
            &= \sum_{a} \inprod{ w_{I^{a0}} }{ w_{I^{a0}} } \\
            &= \inprod{ v_{I} }{ v_{I} }.
    \end{align*}
    Then against $M_{ax}$,
    \begin{align*}
        \inprod{ v_{ M_{ax}} }{ \sum_{b} v_{N_{by}} } &= \inprod{ U_{x}\brac{ w_{I^{ax}} \otimes \ket{a} }  }{ \sum_{b} \sum_{a^{\prime}} w_{ N^{a^{\prime}0}_{by} } \otimes \ket{a^{\prime}} } \\
            &= \inprod{ w_{I^{ax}} \otimes \ket{a}  }{ \sum_{b} U_{x}^{\ast} \brac{ \sum_{a^{\prime}} w_{ N^{a^{\prime}0}_{by} } \otimes \ket{a^{\prime}} } } \\
            &= \inprod{ w_{I^{ax}} \otimes \ket{a} }{ \sum_{b} \sum_{a^{\prime}} w_{ N^{a^{\prime}x}_{by} } \otimes \ket{a^{\prime}} } \\
            &= \inprod{ w_{I^{xa}} }{ \sum_{b} w_{ N^{ax}_{by} } } \\
            &= \inprod{ w_{I^{ax}} }{ w_{I^{ax}} } \\
            &= \inprod{ v_{M_{ax}} }{ v_{I} },
    \end{align*}
    and finally against $N_{b^{\prime}}y^{\prime}$
    \begin{align*}
        \inprod{ v_{N_{b^{\prime}}y^{\prime}} }{ \sum_{b} v_{N_{by}} } &= \inprod{ \sum_{a} w_{N_{b^{\prime} y^{\prime}}^{a0}} \otimes \ket{a} }{ \sum_{b} \sum_{a^{\prime}} w_{ N^{a^{\prime}0}_{by} } \otimes \ket{a^{\prime}} } \\
            &= \sum_{a} \inprod{ w_{N^{a0}_{b^{\prime} y^{\prime}} } }{ \sum_{b} w_{ N_{by}^{a0} } } \\
            &= \sum_{a} \inprod{ w_{N^{a0}_{b^{\prime} y^{\prime}} } }{ w_{I^{a0}} } \\
            &= \inprod{ v_{N_{ b^{\prime} y^{\prime} } } }{ v_{I} }.
    \end{align*}
    
    We then similarly check the constraint $\sum_{a} M_{ax} = I$ against all words in $\mathbf{b}^{1}$. We compute this below for posterity. Against $I$, 
    \begin{align*}
        \inprod{ v_{I} }{ \sum_{a} v_{M_{ax}} } &= \inprod{ \sum_{a^{\prime}} w_{I^{a^{\prime}0}} \otimes \ket{a^{\prime}} }{ \sum_{a} U_{x}\brac{ w_{I^{ax}} \otimes \ket{a} } } \\
            &= \inprod{ \sum_{a^{\prime}} w_{I^{a^{\prime}0}} \otimes \ket{a^{\prime}} }{ \sum_{a} w_{I^{a0}} \otimes \ket{a} } \\
            &= \sum_{a} \inprod{ w_{I^{a0}} }{ w_{I^{a0}} } \\
            &= \inprod{ v_{I} }{ v_{I} }.
    \end{align*}
    Then against $M_{a^{\prime}x^{\prime}}$,
    \begin{align*}
        \inprod{ v_{M_{a^{\prime}x^{\prime} }} }{ \sum_{a} v_{M_{ax}} } &= \inprod{ U_{x^{\prime}}\brac{ w_{I^{a^{\prime}x^{\prime}} } \otimes \ket{a^{\prime}} } }{ \sum_{a} U_{x}\brac{ w_{I^{ax}} \otimes \ket{a} }  } \\
            &= \inprod{ U_{x^{\prime}}\brac{ w_{I^{a^{\prime}x^{\prime}} } \otimes \ket{a^{\prime}} } }{ \sum_{a}  w_{I^{a0}} \otimes \ket{a} } \\
            &= \inprod{  w_{I^{a^{\prime}x^{\prime}} } \otimes \ket{a^{\prime}}  }{ U_{x^{\prime}}^{\ast}\brac{ \sum_{a}  w_{I^{a0}} \otimes \ket{a} } } \\
            &= \inprod{ w_{I^{a^{\prime}x^{\prime}} } \otimes \ket{a^{\prime}} }{ \sum_{a}  w_{I^{ax^{\prime}}} \otimes \ket{a} } \\
            &= \inprod{ w_{I^{a^{\prime}x^{\prime}}} }{ w_{I^{a^{\prime}x^{\prime}}} } \\
            &= \inprod{ v_{M_{a^{\prime}x^{\prime}}} }{ v_{I} },
    \end{align*}
    and against $N_{by}$, 
    \[ \inprod{ v_{N_{by}} }{ \sum_{a} v_{M_{ax}} } = \inprod{ \sum_{a^{\prime}} w_{N_{by}^{a^{\prime}0}} \otimes \ket{a^{\prime}} }{ \sum_{a} w_{I^{a0}} \otimes \ket{a} } = \inprod{v_{N_{by}}}{ v_{I} }. \]
    Hence, $\Lambda$ is a moment matrix for the original NPA hierarchy with the same value as $\Gamma$.
\end{proof}

\printbibliography

\appendix

\section{Examples of nice SoS decompositions}

\subsection{Nice sum-of-squares decomposition of \texorpdfstring{$\Bthree$}{B3}} \label{appsec:niceB3}
In this section, we give a nice SoS decomposition for the $\Bthree$ game which is the $3$-answer generalization of the CHSH game \cite{1911.01593}. The $n$-answer generalization of the CHSH game is a linear constraint satisfaction for the equations
\bearr 
x_0 x_1 &= 1,
x_0 x_1 &= \omega_n,
\eearr 
where $\omega_n = e^{2\pi i/n}$ is the $n$th root of unity. The winning conditions of the $n$-answer CHSH game are 
\bearr 
x=0, y=0 &\Rightarrow a = b,\\
x=1, y=0 &\Rightarrow a = b,\\
x=0, y=1 &\Rightarrow ab = 1,\\
x=1, y=1 &\Rightarrow ab = \omega_n.
\eearr
The $\Bthree$ game is a $3$-answer version from this family of games. It has the following game polynomial,
\be P = A_0 B_0^2 + A_0^2 B_0 + A_0 B_1 + A_0^2 B_1^2 + A_1 B_0^2 + A_1^2 B_0 + \omega^2 A_1 B_1 + \omega A_1^2 B_1^2. \ee
In this paper, we will look at the symmetrized version of the game polynomial with the transformations $B_0^2 \to B_0, \omega^2 \to \omega$ to get the following,
\begin{equation}
    P_{\game{B}_3} = A_0 B_0 + A_0^2 B_0^2 + A_0 B_1 + A_0^2 B_1^2 + A_1 B_0 + A_1^2 B_0^2 + \omega A_1 B_1 + \omega^2 A_1^2 B_1^2,
\end{equation}
where $\omega = \frac{-1 + i\sqrt{3}}{2}$ and $A_0, A_1, B_0, B_1$ are Alice and Bob unitary operators which satisfy the following relations,
\bearr \label{eq:bthree-relations}
A_x^3 = 1 &\implies A_x^\dagger = A_x^2,\\
B_y^3 = 1 &\implies B_y^\dagger = B_y^2,\\
A_x B_y - B_y A_x &= \sbrac{A_x, B_y} = 0.
\eearr
The game polynomial of $\Bthree$ has an optimal quantum value of $6$.
With the relations \eqref{eq:bthree-relations} in mind, we can write the following SoS decomposition for the $\Bthree$ game,
\be 6 - P_{\game{B}_3} = \sum_{i=1}^7 \lambda_i S_i^\dagger S_i, \ee
where the $\lambda_i$'s and $S_i$'s are given by
\bearr 
\lambda_1 = \frac{5}{1872} &,\; S_1 =&& 12 A_0 + \brac{\frac{1}{4} - \omega} B_1 B_0 + \brac{\frac{1}{4} - \omega^2} B_0 B_1 + \brac{\frac{13 \omega}{4} - 7} B_0^2 + \brac{\frac{13 \omega^2}{4} - 7} B_1^2,\\
\lambda_2 = \frac{5}{1872} &,\; S_2 =&& 12 \omega^2 A_1 + \brac{\frac{1}{4} - \omega^2} B_1 B_0 + \brac{\frac{1}{4} - \omega} B_0 B_1 + \brac{\frac{13 \omega}{4} - 7 \omega^2} B_0^2 + \brac{\frac{13 \omega^2}{4} - 7 \omega} B_1^2,\\
\lambda_3 = \frac{5}{4992} &,\; S_3 =&& B_0^2 + \omega B_1^2 + \omega^2 B_0 B_1 + \omega^2 B_1 B_0,\\
\lambda_4 = \frac{1}{11856} &, \; S_4 =&& 114 + (5\omega^2 - 48)A_0B_0 + (5\omega^2 - 23)A_0^2 B_0^2 + (5\omega - 48)A_0 B_1 + (5\omega - 23)A_0^2 B_1^2, \\
\lambda_5 = \frac{259}{1976} &, \; S_5 =&& A_0 B_0 - A_0^2 B_1^2 + \brac{\frac{5 \omega - 3}{7}} (A_0^2 B_0^2 - A_0 B_1),\\
\lambda_6 = \frac{1}{11856} &, \; S_6 =&& 114 + (5 \omega - 48) A_1 B_0 + (5\omega - 23) A_1^2B_0^2 + (5 \omega^2 - 48) (\omega A_1 B_1) + (5 \omega^2 - 23) (\omega^2 A_1^2 B_1^2),\\
\lambda_7 = \frac{259}{1976} &, \; S_7 =&& A_1 B_0 - \omega^2 A_1^2 B_1^2 + \brac{\frac{5 \omega^2 - 3}{7}} (A_1^2 B_0^2 - \omega A_1 B_1).
\eearr 
Note that each $S_i$ term only contains either $A_0$ or $A_1$, which means that the above SoS decomposition is nice as described in our paper. This is an example of a nonlocal NPA level-2 game with non-binary answers for which we can construct a nice SoS decomposition. This gives us hope that the {\niceness} framework is more general than the results in this paper.

\subsection{Nice sum-of-squares decomposition for bipartite matching game}
We will conclude the paper with a direct application of our Theorem~\ref{thm:lvl1-nicelvl1}. Let's take an example of the bipartite matching game $\game{M}$ \cite{matching_game}. In the game setup, the question space is ternary with $\alpbt{X} = \alpbt{Y} = \cbrac{1, 2, 3}$ and the answer space is binary with $\alpbt{A} = \alpbt{B} = \cbrac{0, 1}$. Alice and Bob win the game under the conditions
\bearr 
x = y &\implies a = b,\\
x \neq y &\implies a \neq b.
\eearr
The game polynomial for the bipartite matching game can be written as
\be \gp{M} = A_1 (B_1 - B_2 - B_3) + A_2 (B_2 - B_1 - B_3) + A_3 (B_3 - B_1 - B_2). \ee
This game follows the identities:
\bearr
{A_i}^\dagger = A_i &,\quad {B_i}^\dagger = B_i, \quad \text{Hermiticity}\\
A_i B_j - B_j A_i &= 0, \quad \text{Commutativity}\\
A_i^2 = B_i^2 &= I, \quad \text{Binary operators}
\eearr
The optimal quantum value of the game polynomial $\gp{M}$ is $6$, and \cite{matching_game} present a degree-$1$ SoS certificate for $6 - \gp{M}$. We apply results from Theorem~\ref{thm:lvl1-nicelvl1} to construct a degree-$1$ nice SoS certificate for $6 - \gp{M}$.
\be 6 - \gp{M} = \sum_{i=1}^3 T_i^\dagger T_i, \ee
where the $T_i$'s are given as follows:
\begin{eqnarray}
    T_1 &=& A_1 - \frac{B_1 - B_2 - B_3}{2}\\
    T_2 &=& A_2 - \frac{B_2 - B_1 - B_3}{2}\\
    T_3 &=& A_3 - \frac{B_3 - B_1 - B_2}{2}\\
    T_4 &=& \frac{B_1 + B_2 + B_3}{2}.
\end{eqnarray}
Here, each $T_i$ term only contains either one of the $A_i$, which means that the above SoS decomposition is nice as described in our paper. This is an example of an NPA level-1 game with non-binary questions and binary answers for which we can construct a nice SoS decomposition.

\end{document}